%% file: main.tex
\documentclass[conference]{IEEEtran}
\IEEEoverridecommandlockouts
\usepackage{cite}
\usepackage{fouriernc}
\usepackage{amsmath,amssymb,amsfonts}
\usepackage{amsthm}
\usepackage{multirow}
\usepackage[linesnumbered,ruled]{algorithm2e}
\usepackage{subfigure}
\usepackage{algorithmic}
\usepackage{array}
\usepackage{tabularx}
\usepackage{calrsfs}
\usepackage{graphicx}
\usepackage{cuted}
\usepackage{float}
\usepackage{textcomp}
\usepackage{url}
\usepackage{xcolor}
\usepackage{bm}
\usepackage{comment}
\usepackage{mathtools}
\usepackage{booktabs}       
\usepackage{nicefrac}       
\usepackage{microtype}      
\usepackage{times,epsfig}
\usepackage{balance} 
\usepackage[utf8]{inputenc}
\usepackage{nicefrac}       
\usepackage{microtype}      
\DeclareMathAlphabet\mathbfcal{OMS}{cmsy}{b}{n}
\DeclareMathAlphabet{\mathcal}{OMS}{cmsy}{m}{n}

\newtheorem{theorem}{Theorem}

\usepackage[T1]{fontenc}
\DeclareFontFamily{T1}{calligra}{}
\DeclareFontShape{T1}{calligra}{m}{n}{<->s*[1.44]callig15}{}
\DeclareMathAlphabet\mathcalligra   {T1}{calligra} {m} {n}
\DeclareMathAlphabet\mathzapf       {T1}{pzc} {mb} {it}

\newcount\Comments  
\newcommand{\kibitz}[2]{\ifnum\Comments=1{\color{#1}{#2}}\fi}
\newcommand{\hf}[1]{\kibitz{red}{[hf: #1]}}

\def\BibTeX{{\rm B\kern-.05em{\sc i\kern-.025em b}\kern-.08em
    T\kern-.1667em\lower.7ex\hbox{E}\kern-.125emX}}
\begin{document}

\title{To Warn or Not to Warn: Online Signaling in Audit Games \\
\thanks{This work was supported, in part by grant R01LM10207 from the National Institutes of Health, grants CNS-1526014, CNS-1640624, IIS-1649972 and IIS-1526860 from the National Science Foundation, grant N00014-15-1-2621 from the Office of Naval Research and grant W911NF-16-1-0069 from the Army Research Office.}
}

\author{\IEEEauthorblockN{Chao Yan}
\IEEEauthorblockA{\textit{Dept. of EECS} \\
\textit{Vanderbilt University}\\
Nashville, USA \\
chao.yan@vanderbilt.edu}
\and
\IEEEauthorblockN{~~~~~Haifeng Xu}
\IEEEauthorblockA{\textit{~~~~~Dept. of Computer Science} \\
\textit{~~~~~University of Virginia}\\
~~~~~Charlottesville, USA \\
~~~~~hx4ad@virginia.edu}
\and
\IEEEauthorblockN{Yevgeniy Vorobeychik}
\IEEEauthorblockA{\textit{Dept. of Computer Science and Engineering} \\
\textit{Washington University at St. Louis}\\
St. Louis, USA \\
yvorobeychik@wustl.edu}
\and
\IEEEauthorblockN{Bo Li}
\IEEEauthorblockA{\textit{Dept. of Computer Science} \\
\textit{University of Illinois at Urbana-Champaign}\\
Champaign, USA \\
lbo@illinois.edu}
\and
\IEEEauthorblockN{Daniel Fabbri}
\IEEEauthorblockA{\textit{Dept. of Biomedical Informatics} \\
\textit{Vanderbilt University Medical Center}\\
Nashville, USA \\
daniel.fabbri@vumc.org}
\and
\IEEEauthorblockN{Bradley A. Malin}
\IEEEauthorblockA{\textit{Dept. of Biomedical Informatics} \\
\textit{Vanderbilt University Medical Center}\\
Nashville, USA \\
b.malin@vumc.org}
}

\maketitle

\begin{abstract}
Routine operational use of sensitive data is often governed by law and regulation. For instance, in the medical domain, there are various statues at the state and federal level that dictate who is permitted to work with patients' records and under what conditions.  To screen for potential privacy breaches, logging systems are usually deployed to trigger alerts whenever a suspicious access is detected.  However, such mechanisms are often inefficient because 1) the vast majority of triggered alerts are false positives, 2) small budgets make it unlikely that a real attack will be detected, and 3) attackers can behave strategically, such that traditional auditing mechanisms cannot easily catch them. To improve efficiency, information systems may invoke signaling, so that whenever a suspicious access request occurs, the system can, in \emph{real time}, warn the user that the access may be audited. Then, at the close of a finite period, a selected subset of suspicious accesses are audited.  This gives rise to an online problem in which one needs to determine 1) whether a warning should be triggered and 2) the likelihood that the data request event will be audited. In this paper, we formalize this auditing problem as a Signaling Audit Game (SAG), in which we model the interactions between an auditor and an attacker in the context of signaling and the usability cost is represented as a factor of the auditor's payoff. We study the properties of its Stackelberg equilibria and develop a scalable approach to compute its solution. We show that a strategic presentation of warnings adds value in that SAGs realize significantly higher utility for the auditor than systems without signaling.  We perform a series of experiments with 10 million real access events, containing over 26K alerts, from a large academic medical center to illustrate the value of the proposed auditing model and the consistency of its advantages over existing baseline methods.
\end{abstract}

\begin{IEEEkeywords}
database auditing, privacy, signaling, Stackelburg game
\end{IEEEkeywords}

\input{I_introduction.tex}

\input{III_method.tex}
\input{IV_Opt_SAG.tex}

\input{V_theory.tex}

\input{VI_evaluation.tex}

\input{II_literature.tex}
\input{VII_discussion.tex}


\bibliographystyle{IEEEtran}
\bibliography{reference}

\end{document}

%% file: I_introduction.tex
\section{Introduction}

Our society now collects, stores, and processes personal and intimate data with ever-finer detail, documenting our activities and innovations in a wide range of domains, ranging from health to finance \cite{mcafee2012big,yin2015big}.  Due to the potential value of such data, their management systems face non-trivial challenges to personal privacy and organizational secrecy. The sensitive nature of the data stored in such systems attracts malicious attackers who can gain value by disrupting them in various ways (e.g., stealing sensitive information, commandeering computational resources, committing financial fraud, and simply shutting the system down) \cite{pavlou2012dragoon,fabbri2013select}. Reports in the popular media  indicate that the severity and frequency of attack events continues to grow. Notably, the recent breach at Equifax led to the exposure of data on $143$ million Americans, including credit card numbers, Social Security numbers, and other information that could be used for identity theft or other illicit purposes \cite{cnn2017equifax}.  Even more of a concern is that the exploit of the system continued for at least two months before it was discovered.

To defend against attack, modern database systems are often armed with an alerting capability to detect and notify about potential risks incurred during daily use \cite{terzi2015survey,hasan2016secure,puppala2016data}.
This entails the logging of access events, which can be thought of as a collection of rules, each of which defines a semantic type of a potentially malicious situation \cite{motwani2008auditing,mazzawi2017anomaly}. 
In mission-critical systems, the access requests of authenticated users are often granted to ensure continuity of workflow and operations, such that notification about potential misuse is provided to administrators who perform retrospective audit investigations \cite{kuna2014outlier,blocki2012audit,lu2009auditing,groomer2018continuous}.
For instance, many healthcare organizations (HCOs) rely on alert, as well auditing, mechanisms to monitor anomalous accesses to electronic medical records (EMRs) by employees who may violate policy and breach the privacy of certain patients \cite{hedda2017evaluating}.
Similarly, the providers of online services, such as financial institutions and social media platforms, often use alerts and audits to defend against attacks, such as financial fraud 
and compromises to computational resources \cite{barth2007privacy}.
Though audits do not directly prevent attacks in their own right, they allow for the discovery of breaches that can be followed up on before they escalate to full blown exploits by attackers.

However, there are challenges to instituting robust auditing schemes in practice. First, the volume of triggered alerts is typically far greater than the auditing capacity of an organization \cite{laszka2017game}. Second, in practice, the majority of triggered alerts correspond to false positives, which stem from an organization's inability to define and recognize complex dynamic workflows.
Third, to mitigate the risk of being caught, attackers prefer to act strategically, such as carefully choosing the way (or target) to attack. And last, but not least, in the retrospective audit setting, attacks are not discovered until they are investigated. 

In essence, this is a resource allocation problem in an adversarial environment for which the Stackelberg security game (SSG) is a natural choice to apply for modeling  purposes \cite{tambe2011security,fang2016deploying,do2017game,sinha2018stackelberg}. In this model, the \emph{defender} first commits to a budget allocation policy and, subsequently, the \emph{attacker} responds with the optimal attack based on the defender's strategy. This model has enabled the design and deployment of solutions to various security problems in practice, such as \emph{ARMOR} (which was adopted by the LAPD to randomize checkpoints on the roadways at Los Angeles International Airport) \cite{pita2009using} and \emph{IRIS} (which was adopted by the US Federal Air Marshal Service to schedule air marshals on international flights) \cite{jain2010software}.
The audit game is a variation of the SSG designed to discover an efficient audit strategy  \cite{blocki2013audit,blocki2015audit,yan2018get,yan2019database}. With respect to strategic auditing, most research has focused on deriving a defense strategy by solving, or approximating, the Strong Stackelberg Equilibrium (SSE).
Unfortunately, it was recently shown that merely applying the SSE strategy may have limited efficacy in some security settings \cite{xu2015exploring}.  
This can be addressed by strategically revealing information to the attacker \cite{xu2015exploring,rabinovich2015information}, a mechanism referred to as \emph{signaling} (or \emph{persuasion} \cite{kamenica2011bayesian,dughmi2016algorithmic}). In this setting, the goal is to  set up  a \emph{signaling scheme} to reveal  noisy information  to the attacker and, by doing so, influence the attacker's decision with respect to outcomes that favors the defender.
However, all approaches derived to date rely on allocating resources \emph{before} signaling, such that it serves as a source of informational advantages for deceiving the attacker. 
Yet, in the audit setting, the decision sequence is reversed, such that the signal is revealed (e.g., via a warning screen) at the time of an access request, whereas the audit occurs after a certain period of time. This poses new challenges for the design of signaling schemes. 
 
Many organizations have recognized and adopted signaling mechanisms to protect sensitive data.
For example, in 2018, Vanderbilt University Medical Center (VUMC) announced a new \emph{break-the-glass} policy to protect the privacy of patients with a \emph{person of interest} (or VIP) designation, such as celebrities or public figures.\footnote{https://www.mc.vanderbilt.edu/myvumc/index.html?article=21557}
 Under this policy, access to the EMRs of these individuals triggers a pop-up warning that requires the user to provide a justification for the access. Once the warning has been served, the user can decide whether or not to proceed to access, knowing that each access is logged for potential auditing. However, such a policy is implemented in a \emph{post hoc} manner that does not optimize when to signal nor when to audit.

In this paper, we introduce the notion of a Signaling Audit Game (SAG), which applies signaling to alerts and auditing. We  leverage  the time gap between the access request made by the (potential) attacker and the actual execution of the attack to insert the signaling mechanism. When an alert is triggered by a suspicious access request, the system can, in real time, send a warning to the requestor. At this point, the attacker has an opportunity to re-evaluate his/her utility and make a decision about whether or not to continue with an attack. In contrast to previous models, which are all computed offline, the SAG optimizes both the warning strategy and the audit decision in real time for each incoming alert.
Importantly, we consider the usability cost into the SAG where the normal data requestors may be scared away by the warning messages in practice. This may lead to descent in operational efficiency of organizations which deploy SAGs.
To illustrate the performance of the SAG, in this paper we evaluate the expected utility of the auditor with a dataset of over 10 million real VUMC EMR accesses and predefined alert types.
The results of a comprehensive comparison, which is  performed  over a range of conditions, indicate that the SAG consistently outperforms state-of-the-art game theoretic alternatives that lack signaling by achieving higher overall utility while inducing nominal increases in computational burden.

The remainder of this paper is organized as follows. We  first propose the SAG and introduce how it is played in the audit setting. Next, we analyze the theoretical properties of  the SAG equilibria. The dataset, experiments, and results are then described in the evaluation section. Finally, we review representative related research in the database auditing domain, with a focus on methodology in the adversarial setting.

%% file: III_method.tex
\section{Online Signaling in Audit Games}
In this section, we describe the SAG model in the general context of information services. For illustrative purposes, we use healthcare auditing as a running example.

\subsection{Motivating Domain}

To provide efficient healthcare service, HCOs typically store and process each patient's clinical, demographic, and financial information in an EMR system. 
EMR users, such as physicians and other clinical staff, need to access patients' EMRs when providing healthcare services. 
The routine workflow can be summarized as three steps: 1) a user initiates a search for a patient's EMR by name and date of birth, then the system returns a list of patients (often based on a fuzzy matching) along with their demographic information, 2) from the list, this user requests access to a patient's record, and 3) the system returns the requested record.  Due to the complex, dynamic and time-sensitive nature of healthcare, HCOs typically grant employees broad access privileges, which unfortunately creates an opportunity for malicious insiders to exploit patients' EMRs \cite{gunter2011experience}.

To deter malicious access, breach detection tools are commonly deployed to trigger alerts in real time for the administrator whenever suspicious events occur.
Alerts are often marked with predefined types of potential violations which help streamline inspection. Notable alert types include accessing the EMR of co-workers, neighbors, family members, and VIPs \cite{hedda2017evaluating}. 
Subsequently, a subset of the alerts are retrospectively audited at the end of each audit cycle, and the auditor determines which constitute an actual policy violation.

\subsection{Signaling Audit Games}
Here, we formalize the  Signaling Auditing Game (SAG) model. An SAG is played between an \emph{auditor} and an \emph{attacker} within a predefined audit cycle (e.g., one day). 
This game is sequential such that alerts arrive one at a time.
For each alert, the auditor needs to make two decisions in \emph{real time}: first, which signal to send (e.g., to warn the user/attacker or not), and second, whether to audit the alert.
Formally, let $X_c^{\tau}$ denote the event that alert $\tau$ will be audited, and $X_u^{\tau}$ denote that it is not audited.
We further let $\xi_1^{\tau}$ denote the event that a \emph{warning signal} is sent for alert $\tau$, 
while $\xi_0^{\tau}$ denotes the event that no warning is sent (i.e. a ``silent signal''). The warning $\xi_1^{\tau}$ is delivered privately through a dialog box on the requestor's screen, which might communicate ``\emph{Your access may be investigated. Would you like to proceed?}''.  $X_c^{\tau},X_u^{\tau} , \xi_1^{\tau}, \xi_1^{\tau}$ are random variables whose probabilities are to be designated.

We assume that there is a finite set of alert types $T$ and, for each $t \in T$, all alerts  are considered equivalent for our purposes (i.e., attacks triggering alerts of type $t$ all result in the same damages to the system).
The auditor has an auditing budget $B$ that limits the number of alerts that can be audited at the end of the cycle.
For each alert type $t$, let $V^t$ denote the cost (or time needed) to audit an alert of type $t$.
Thus, if $\theta^t$ is the probability of auditing alerts of type $t$ and $d^t$ is the number of such alerts, the budget constraint implies that $\sum_t \theta^t \cdot V^t d^t \le B$.

Since the setting is online, an optimal policy for the auditor must consider all possible histories of alerts, including the correlation between alerts.
Given that this is impractical, we simplify the scheme so that 1) each alert is viewed independently of alerts that precede it and 2) future alerts are considered with respect to their average relative frequency. 
Specifically, we assume that each attack effectively selects an alert type $t$, but do not need to consider the timing of attacks. Rather, we treat each alert as potentially adversarial.
This implicitly assumes that an attack (e.g., a physician's access to the EMR of a patient they do not treat) triggers a single alert.
However, this is without loss of generality, since we can define alert types that capture all realistic multi-alert combinations.


Now, we define the payoffs to the auditor and attacker.
For convenience, we refer to the alert corresponding to an attack as the \emph{victim alert}. 
If the auditor fails to audit a victim alert of type $t$,  the auditor and the attacker will receive utility $U^t_{d,u}$ and $U^t_{a,u}$, respectively. On the other hand, if the auditor audits a victim alert of type $t$, the auditor and the attacker will receive utility $U^t_{d,c}$ and $U^t_{a,c}$, respectively.
Here, the subscripts $c$ and $u$ stand for \emph{covered} and \emph{uncovered}, respectively. Naturally, we assume $U^{t}_{a,c} < 0 < U^{t}_{a,u}$ and $U^{t}_{d,c} \ge 0 > U^{t}_{d,u}$. 

Figure \ref{fig1} demonstrates the key interactions of both players along the timeline. Each yellow block within the audit cycle represents a triggered alert and the corresponding interactions with it. At a high level, the SAG consists of an online component and an offline component, as shown in the figure. In the online component, the auditor continues to update the real time probability of auditing any alert (may or may not be triggered) with respect to the alert type and the time point $\tau$. In other words, the auditor commits in real time to the auditing and signaling strategy. In this case, the auditor always moves first, as shown at the beginning of the lower timeline. 
\begin{figure}[h]
	\centering
	\includegraphics[width=7cm]{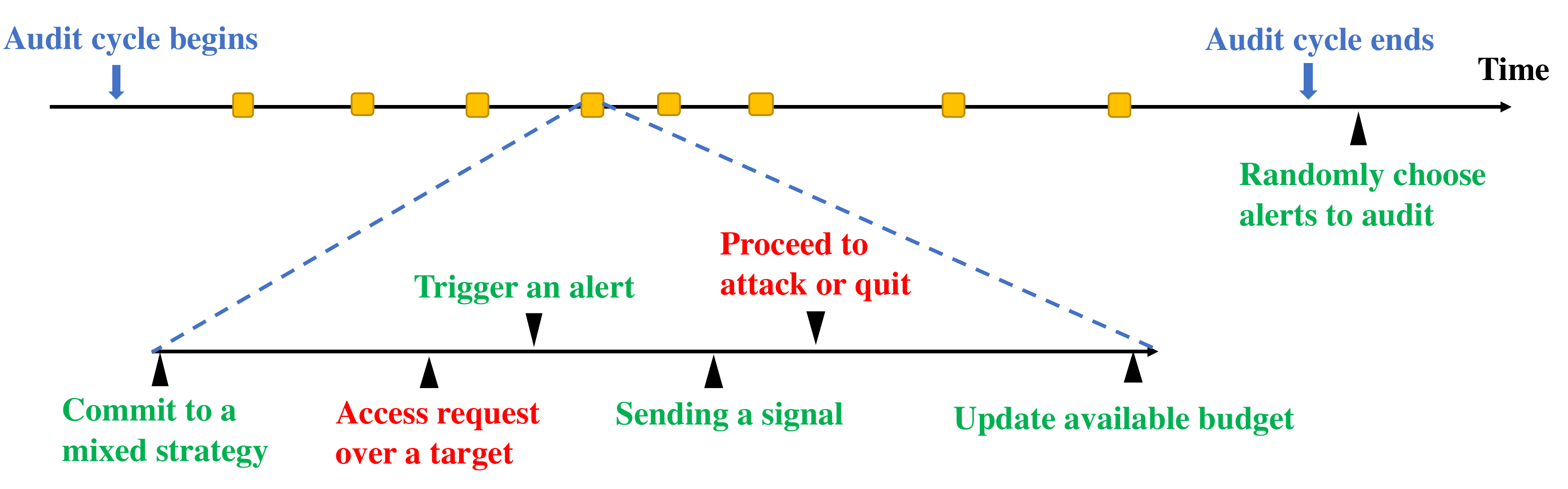}
	\caption{The auditor and attacker actions are shown in \emph{green} and \emph{red}, respectively.  }
	\label{fig1}
\end{figure}

A \emph{warning signaling scheme}, captured by the joint probability distribution of signaling and auditing, can be fully specified through four variables for each $\tau$:
\begin{equation}\label{audit_scheme}
\begin{aligned}
\mathbf{P}(\xi_1^{\tau}, X_c^{\tau}) = p_1^{\tau}, ~~~ \mathbf{P}(\xi_1^{\tau}, X_u^{\tau}) = q_1^{\tau}, \\
\mathbf{P}(\xi_0^{\tau}, X_c^{\tau}) = p_0^{\tau}, ~~~ \mathbf{P}(\xi_0^{\tau}, X_u^{\tau}) = q_0^{\tau}.\\
\end{aligned}
\end{equation}  
Upon receiving the signal, the attacker reacts as follows:
\begin{itemize}
	\item  
	After $\xi_1^{\tau}$: the system presents two choices to the attacker: \emph{``Proceed''} to access the requested record or quit.  
	\item  After $\xi_0^{\tau}$: the attacker automatically \emph{proceeds} to access the requested record (since the attacker receives no warning).
\end{itemize}
For convenience, when possible we omit the superscript $\tau$  when the alert we are dealing with, is readily apparent from the context. 

Figure \ref{fig0} illustrates the temporal sequence of decisions in the SAG.
Each edge in the figure is marked with its corresponding joint probability of a sequence of decisions up to and including that edge. 
Note that the two gray nodes are not extended because they do not lead to any subsequent event.\footnote{The upper gray node corresponds to the case when an access request is abandoned. 
	The lower one represents an impossible case because the user automatically gets the requested record.}
\begin{figure}[h]
	\centering
	\includegraphics[width=6cm]{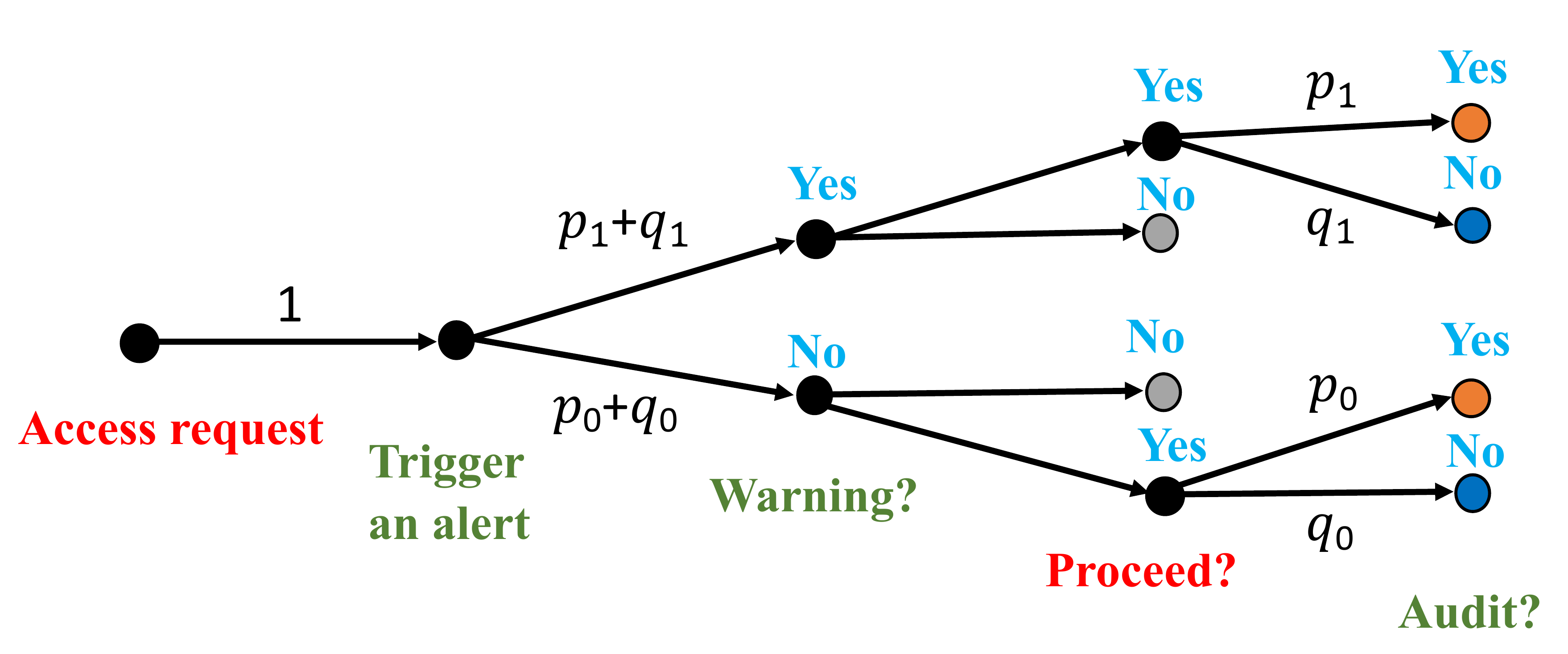}
	\caption{The decision tree of the auditor and an arbitrary user, the actions for which are shown in \emph{green} and \emph{red}, respectively.}
	\label{fig0}
\end{figure}
Further, observe that, $p_1 + q_1 + p_0 + q_0 = 1$, and the overall probability of auditing this alert  is $\mathbf{P}(X_c) = \mathbf{P}(X_c, \xi_1) + \mathbf{P}(X_c,\xi_0) = p_1 + p_0$. Conditional on the warning signal $\xi_1$, the probability of auditing this alert is thus  $\mathbf{P}( X_c | \xi_1) = p_1/(p_1+q_1)$. 

Since the auditor has a fixed auditing budget, she will need to update the remaining budget after determining the signal-conditional audit probability for the current alert. We use $B_{\tau}$ to denote the remaining budget \emph{before} receiving alert $\tau$. Let $t$ denote the type of alert $\tau$ and $\tau+1$ denote the next alert. After the signaling scheme for $\tau$ is executed, the auditor then updates $B_{\tau}$ for the use of the next alert $\tau+1$ as follows: 
\begin{itemize}\label{B_update_1}
	\item If $\xi_1^{\tau}$ is sampled: $B_{\tau + 1}  =   B_{\tau} - p_1^{\tau}/(p_1^{\tau} +q_1^{\tau}) \cdot V^t$. 
	\item If $\xi_0^{\tau}$ is sampled: $B_{\tau + 1 }  =   B_{\tau} - p_0^{\tau}/(p_0^{\tau} +q_0^{\tau}) \cdot V^t$. 
\end{itemize}
Additionally, we always ensure that $B_{\tau} \geq 0$.
The key challenge in our model is to compute the optimal $p_1^{\tau}, q_1^{\tau}, p_0^{\tau}, q_0^{\tau}$ for each alert $\tau$ \emph{online} by accounting for the remaining budget and the estimate number of future alerts. This needs to be performed to ensure that the auditor does not spend the budget at a rate that is excessively fast or slow. 

Without signaling, our audit game can be solved \emph{offline}, at the end of the audit cycle.
This situation can be captured by a  Stackelberg security game by viewing alerts as targets \cite{tambe2011security}.
The optimal auditing probabilities can then be determined offline by computing the SSE of this game. However, as our experiments show, this simplified strategy (which we refer to as \emph{offline SSE}) performs substantially worse than our online approach.  

The SAG can be viewed as a variation on the Stackelberg game, where it includes signaling and makes decisions about auditing \emph{online} upon the arrival of each alert.
The premise behind our solution is therefore a Strong Stackelberg equilibrium of the SAG, in which the auditor commits to a randomized joint signaling and auditing decision, and the associated probability distribution is observed by the attacker, who then decides first upon the alert type to use, and subsequently whether to proceed after a warning.
We will seek the optimal randomized commitment strategy for the auditor in this game.


The SAG model contains two crucial differences from prior investigations into signaling for security games. The first is that the signaling scheme for each alert in an SAG must be optimized sequentially in \emph{real time}. By contrast, previous models, such as \cite{xu2015exploring}, decide the signaling schemes for all targets simultaneously in an offline fashion.  The second is in how private information is leveraged. In previous models, the defender utilizes the informational advantage that the defender \emph{currently} has (e.g., knowledge about the realized protection status of the target) to deceive the attacker. However, in our scenario, the auditor first decides the signaling scheme, by when he/she has an equal amount of information as the attacker (which includes the  status of the current environment), and then exercises her informational advantage \emph{after} the audit cycle ends (by deciding which to audit).

%% file: IV_Opt_SAG.tex
\section{Optimizing SAGs}\label{sec:alg}

In this section, we design an algorithm for solving SAGs. For presentation purpose, we  fix the alert $\tau$ to a particular type $t$ and, thus, the superscript will, at times, be omitted for notational convenience. We begin by considering the problem of computing the real time SSE of the game without signaling that transpires for a given observed alert $\tau$. This game, as well as its solution, serve as a baseline of the optimized SAGs.

\subsection{Online SSG}\label{online_SSE}


Consider the arrival of an alert $\tau$.
Let $d^{t}_{\tau}$ be the number of future alerts of type $t\in T$ after alert $\tau$ is triggered.\footnote{The vast majority of alerts are false positives. Consequently, we can estimate $d^{t}_{\tau}$ from alert logs in previous audit cycles.}
We assume that $d^{t}_{\tau}$ follows a Poisson distribution $\bm D^{t}_{\tau}$, which is widely adopted to characterize the number of arrivals.
We can compute the SSE strategy using a multiple linear programming (LP) approach for budget $B_\tau$. In this approach, for each alert type $t$, we assume that $t$ is the attacker's best response, and then compute the optimal auditing strategy.
Finally, we choose the best solution (in terms of  the auditor's utility) among all of the LPs  as the SSE strategy.

Now, let $\theta^{t'}(t)$ be the probability of auditing an alert of type $t'$ when the attacker's best response is $t$.  
In addition to this optimal auditing policy, we design how we plan to split the remaining budget $B_\tau$ among all alert types.
We assume that the audit distribution will remain constant for future alerts, which allows us to consider the long-term impact of our decision about auditing. 
We represent the budget that we allocate for inspecting alerts of each type as a vector  $\bm B_{\tau} = \{B^{1}_{\tau},B^{2}_{\tau},...,B^{|T|}_{\tau}\}$ that the long-term budget allocation decision is constrained by the remaining audit budget: $\sum_{t=1}^{|T|} B^{t}_{\tau} \le B_{\tau}$. Now, assuming type $t$ is the best response, the following LP returns optimal auditing strategy:
\begin{equation}\label{SSG_SSE}
\begin{array}{l}
\quad \quad \max_{\bm B_{\tau}}   \theta^{t}(t)\cdot U^{t}_{d,c} + (1-\theta^{t}(t) ) \cdot U^{t}_{d,u}\\
 s.t.  \\
 \quad \quad \quad \forall t',   \quad  \theta^{t}(t)  \cdot  U^{t}_{a,c} + (1-\theta^{t}(t) ) \cdot U^{t}_{a,u}  \\[1pt]
 \quad \quad  \quad \quad \quad \quad  \quad \quad  \ge \theta^{t'}(t) \cdot U^{t'}_{a,c} + (1-\theta^{t'}(t)  ) \cdot U^{t'}_{a,u} ,\\[5pt]
 \quad \quad  \quad \forall t',   \quad \theta^{t'}(t) = \mathbb{E}_{d^{t'}_{\tau} \sim \bm D^{t'}_{\tau}}\left(\frac{B^{t'}_{\tau} }{ V^{t'} d^{t'}_{\tau}}\right) ,\\[5pt]
\quad  \quad  \quad  \quad \quad \quad \sum_{t'=1}^{|T|} B^{t'}_{\tau} \le B_{\tau} , \quad  \quad  \quad  \quad\\[5pt]
 \quad  \quad \quad  \forall t', \quad B^{t'}_{\tau} \in [0, B_{\tau}],\\
\end{array}
\end{equation}
where the first constraint ensures that $t$ is the attacker's best response.
After solving $|T|$ instances of LP \eqref{SSG_SSE}, the best solution for the auditor will henceforth be referred to as the \emph{online SSE strategy} (or simply, the \emph{SSE}), $\bm{\theta}_{SSE}$.

\subsection{Optimal Signaling}

We now describe how to build a signaling mechanism into the audit game and then compute the optimal signaling scheme, as well as the budget allocation strategy.

From the perspective of the attacker, whether to \emph{proceed} or \emph{quit} after receiving a warning signal depends on his conditional expected utility: $$\mathbb{E}^{t}_a(util  | \xi_1) = \frac{p^t_1}{p^t_1+q^t_1} \cdot U^{t}_{a,c} + \frac{q^t_1}{p^t_1+q^t_1} \cdot U^{t}_{a,u}. $$  
We impose the constraint $\mathbb{E}^{t}_a(util  | \xi_1) \le 0$ such that the attacker's best response to $\xi_1$ is to quit, in which case both players will receive $0$ utility. We do not enforce constraints for $\xi_0$ because the potential attacker does not have any option but to proceed. In this case, the expected utility of the auditor is  
$$\mathbb{E}^{t}_d(util  | \xi_0) = \frac{p^t_0}{p^t_0+q^t_0} \cdot U^{t}_{d,c} + \frac{q^t_0}{p^t_0+q^t_0} \cdot U^{t}_{d,u}. $$   
Overall, the expected utility for the attacker can be computed as  $$\mathbb{E}^{t}_a(util ) = (p^t_0 + q^t_0) \cdot \mathbb{E}^{t}_a(util |\xi_0) =  p^t_0 \cdot U^{t}_{a,c} + q^t_0 \cdot U^{t}_{a,u}.$$ Accordingly, the auditor's expected utility is $$\mathbb{E}^{t}_d(util ) = (p^t_0 + q^t_0) \cdot \mathbb{E}^{t}_d(util |\xi_0) =  p^t_0 \cdot U^{t}_{d,c} + q^t_0 \cdot U^{t}_{d,u}.$$ 

However, a side effect is that, the warnings sent by the auditor (e.g., the pop-up warning screen off of \emph{break-the-glass} strategy deployed by VUMC) may pose an additional utility loss to the auditor in practice, which we call \emph{usability cost}. This is because when normal users request access to sensitive data and receive a warning message, they may walk away by choosing quit instead of ``Proceed'', which induces a loss in operational efficiency for the organization.  For each type $t'$, we set this loss to be proportional to the product of the probability of sending warnings $p^{t'}_1 + q^{t'}_1$, the probability of being deterred $P^{t'}$ and the expectation of the number of future alerts to the end of the current audit cycle $E^{t'}_{\tau} =\mathbb{E}_{d^{t'}_{\tau} \sim \bm D^{t'}_{\tau}}(d^{t'}_{\tau})$, which can be estimated from historical data collected in previous audit cycles. The loss incurred for each quit by a normal user is set to be $C_{t'} (<0)$. Then, the expected utility of the auditor can be updated as $\mathbb{E}^{t}_d(util) = p^t_0 \cdot U^{t}_{d,c} + q^t_0 \cdot  U^{t}_{d,u} + \sum_{t'=1}^{|T|} (p^{t'}_1 + q^{t'}_1) \cdot P^{t'} \cdot E^{t'}_{\tau} \cdot C_{t'}$. Note that here we use $E^{t'}_{\tau}$ to estimate the  number of false-positive alerts in the future. This is reasonable because attacks are extremely rare in practice.


The optimal signaling scheme (or, more concretely, joint signaling and audit probabilities) can be computed through the following set of LPs:
\begin{equation}\label{basic_game}
\begin{aligned}
\max\limits_{ \bm{p_0, p_1, q_0, q_1}, \bm B_{\tau}}  p^t_0 \cdot U^{t}_{d,c} + q^t_0 \cdot  U^{t}_{d,u} + \sum_{t'=1}^{|T|} (p^{t'}_1 + q^{t'}_1) \cdot P^{t'} \cdot E^{t'}_{\tau} \cdot C_{t'} \\[-3pt]
s.t.  ~~~~~~~~~~~~~~~~~~~~~~~~~~~~~~~~~~~~~~~~~~~~~~~~~~~~~~~~~~~~~~~~~~~   \\[-3pt]
\quad \quad  \forall t',   \quad  p^{t}_0  \cdot  U^{t}_{a,c} + q^{t}_0  \cdot U^{t}_{a,u} \ge p^{t'}_0 \cdot U^{t'}_{a,c} + q^{t'}_0 \cdot U^{t'}_{a,u}, ~~~~~~\\[1pt]
 \quad \quad \forall t',   \quad p^{t'}_1  \cdot  U^{t'}_{a,c} + q^{t'}_1  \cdot U^{t'}_{a,u} \le  0,~~~~~~~~~~~~~~~~~~~~~~~~~~\\[1pt]
 \quad \quad \forall t',   \quad p^{t'}_1  + p^{t'}_0  = \mathbb{E}_{d^{t'}_{\tau} \sim \bm D^{t'}_{\tau}}\left(\frac{B^{t'}_{\tau} }{ V^{t'} d^{t'}_{\tau}}\right), ~~~~~~~~~~~~~~~~~~~~\\[-3pt]
 \quad \quad \forall t',   \quad p^{t'}_1  + p^{t'}_0 + q^{t'}_1  + q^{t'}_0  = 1, ~~~~~~~~~~~~~~~~~~~~~~~~~~~\\[0pt]
\quad  \quad   \sum_{t'\in \{1,...,|T|\}} B^{t'}_{\tau} \le B_{\tau}, ~~~~~~~~~~~~~~~~~~~~~~~~~~~~~~~~~~~ \\[1pt]
\quad  \quad  \forall t',  \quad B^{t'}_{\tau} \in [0, B_{\tau}], ~~~~~~~~~~~~~~~~~~~~~~~~~~~~~~~~~~~~~~\\[1pt]
\quad  \quad  \forall t',  \quad p^{t'}_0, q^{t'}_0, p^{t'}_1, q^{t'}_1 \in [0,1], ~~~~~~~~~~~~~~~~~~~~~~~~~~~\\
\end{aligned}
\end{equation}
where we assume type $t$ is the best one for the attacker to potentially exploit. Note that, in the objective function, the incurred additional loss is an accumulated value that considers the amount of time remaining in the period for the current audit cycle. The likelihood of sending warning signal in the current time point is a real time estimation of future warnings. Our goal is thus to find the optimal signaling scheme for all types, and simultaneously, the best budget allocation strategy. We use $\bm{p_0}$,  $\bm{p_1}$, $\bm{q_0}$ and $\bm{q_1}$ to denote the warning signaling scheme for all types, namely, the set $\{p_0^{t'} | \forall t'\}$, $\{p_1^{t'} | \forall t'\}$, $\{q_0^{t'} | \forall t'\}$ and $\{q_1^{t'} | \forall t'\}$, respectively. 

The first constraint in LP \eqref{basic_game} ensures that attacking type $t$ is the best response strategy for the attacker. The second constraint indicates that the attacker, when receiving a warning signal, will quit attacking any type. 
We refer to the optimal solution among the $|T|$ instances of LP \eqref{basic_game}  as the \emph{Online Stackelberg Signaling Policy (OSSP)}. In particular, we use $\bm{\theta}_{ossp}$ to denote the vector of coverage probability at OSSP.

After building the theoretical model of the SAG, we need to pay attention to one important situation in practice, where an attacker can leverage to perform attacks with lower level risks of being captured.

\subsection{The Ending Period of Audit Cycles}

Recall that in SAGs, the estimation of the number of alerts in the rest of the current audit cycle, which is $\mathbb{E}_{d^{t}_{\tau} \sim \bm D^{t}_{\tau}}(d^{t}_{\tau})$, is calculated based on the alert logs of historical audit cycles. At the ending period of audit cycles, such estimation keeps decreasing for each type.  As a consequence, it would be ill-advised to apply any approach that performs an estimation on the arrivals without an additional process to handle the ending period of an audit cycle. Imagine, for instance, an attacker who only attacks at the very end of an audit cycle. Then, the knowledge from historical data is likely to indicate that no alerts will be realized in the future. And it follows that such attacks will not be covered because the available budget will have been exhausted according to the historical information. 

To practically mitigate this problem, when the mean of arrivals in the historical data drops under a certain threshold, we apply the estimate of the number of future alerts $\mathbb{E}_{d^{t}_{\tau-1} \sim \bm D^{t}_{\tau-1}}(d^{t}_{\tau-1})$ in the time point when the last alert was triggered as a proxy of the real one at the current time point. This technique is called \emph{knowledge rollback}. By doing so, the consumption of the available budget in real time will be slowed down because of the application of a smaller coverage probability. As a consequence, the attacker attempting to attack late is not afforded an obvious extra benefit.

%% file: V_theory.tex
\section{Theoretical Properties of SAGS}

In this section, we theoretically analyze the properties of the OSSP solution (equivalently, of the SAG equilibrium). Our first result highlights a notable property of the optimal signaling scheme. Specifically, the optimal signaling scheme will only trigger warning signals for the best attacking type, i.e., the type at which attacker utility is maximized. As such, the rational attacker will choose to attack this alert type. 

\begin{theorem}\label{sig_zero}
If alert $\tau_{*}$ of type $t_{*}$ is the best response strategy for the attacker, then  $p_1^{t}=q_1^{t}=0$ in the OSSP for $\forall t \neq t_{*}$.\footnote{We will use $*$ to denote strategies or quantities in the OSSP in the rest of the paper.}
\end{theorem}
\begin{proof}

	Let $Sol = \{ p_0^t, p_1^t, q_0^t, q_1^t \}_{t \in T}$ be any optimal solution and $t_*$ is the best type. We show that the following newly defined variables will not decrease the objective value of $Sol$ and thus, by assumption, is still optimal. Let $\bar{p}_0^{t_*} = p_0^{t_*} , \bar{p}_1^{t_*} = p_1^{t_*}, \bar{q}_0^{t_*} = q_0^{t_*}, \bar{q}_0^{t_*} = q_0^{t_*}$  be the same as in $Sol$, however for any $t \not = t_*$, define $\bar{p}_0^{t} = p_0^{t} + p_1^t , \bar{q}_0^{t} = q_0^{t} + q_1^t$ and $ \bar{p}_1^{t} = 0, \bar{q}_1^{t} = 0$. 
	
	First, we argue that these newly defined variables are still feasible. All of the constraints can easily be verified in LP \eqref{basic_game} except the first two sets. The second set of constraints is still satisfied for any $t \not = t_*$ (where our variables changed) since $\bar{p}_1^t = \bar{q}_1^t = 0$. The first set of constraints are satisfied for any $t \not = t_*$ because 	
	\begin{eqnarray*}
		\bar{p}_0^t\cdot U_{a,c}^t + \bar{q}_0^t \cdot U_{a,u}^t &=&  (p_0^t + p_1^t)\cdot U_{a,c}^t + (q_0^t + q_1^t) \cdot U_{a,u}^t  \\
		& \leq & p_0^{t_*} \cdot U_{a,c}^t + q_0^{t_*}  \cdot U_{a,u}^t  \\
		& = & \bar{p}_0^{t_*} U_{a,c}^t + \bar{q}_0^{t_*} U_{a,u}^{t_*} 
	\end{eqnarray*}
	where the (only) inequality is due to $ p_1^t \cdot U_{a,c}^t + q_1^t \cdot U_{a,u}^t \leq 0$ as a constraint of LP \eqref{basic_game} and the two equations are by our definition of the new variables. This proves that the first constraint is also feasible.
	
	It remains to show that the newly defined variables do not decrease the objective function. This follows simply because the term with respect to type $t_*$ in the objective function does not change and all the other terms become zero in the newly defined variables, which is no less than the original cost. This proves the theorem. \qedhere

\end{proof}


Theorem \ref{sig_zero} leads to the following corollary:
when the attacker avoids attacking certain type(s) at any time point (this is always the case in OSSP), then the best strategy for the auditor is to turn off the signaling procedure for those types for less loss incurred by sending warnings. Now we show that, at any given game status, the marginal coverage probability for OSSP is the same as the one for the online SSE.

\begin{theorem}\label{SSE_peSSE_equivalence}
Let $\theta^t_{ossp}$ be the marginal coverage probability in the OSSP at any given game status and $\theta^t_{SSE}$ be the corresponding marginal coverage probability in the online SSE. Then, in a SAG, for each type $t\in T$,  
$\theta^t_{ossp} = \theta^t_{SSE}$. 
\end{theorem}
\begin{proof}
%
Given any game state, the auditor has an estimate about the sets of future alerts. We prove that for any fixed set of alerts,  $\theta^t_{ossp} = \theta^t_{SSE}$ holds for each type $t\in T$. As a result, in expectation over the probabilistic estimate, this still holds.  

Fixing a set of alerts, the auditor's decision is a standard Stackelberg game. 
We first claim that by fixing the auditing strategy in the OSSP, the attacker can receive $\mathbb{E}_a^{ossp}$ by triggering any alert $\tau$, thus type $t$. In other words, $\forall t \neq t_*, \mathbb{E}_a(\theta^{t}_{ossp})=\mathbb{E}_a^{ossp}$. Assume, for the sake of contradiction, that an alert $\tau'$ of type $t'$ with positive coverage probability is not the best response of the attacker in an SAG. Then, the auditor can redistribute a certain amount of the protection resources from $\tau'$ to the alerts of the attacker's best-response type and guarantee that it is still the best-response type. This increases the coverage probability of these alerts and, thus, increases the auditor's utility, which contradicts the optimality of OSSP. This implies that the first constraint in LP \eqref{basic_game} is \emph{tight} in the OSSP. Similarly, this holds true for the online SSE. Notice that $\mathbb{E}_a(\theta^t)$  is a strictly decreasing function of $\theta^t$ for both OSSP and online SSE.


Next, we prove that $ \mathbb{E}_a^{sse} = \mathbb{E}_a^{ossp}$ implies $\theta^t_{ossp} = \theta^t_{SSE}$ for all $\tau$, thus $t$, as desired.  
This is because $\theta^t_{ossp} >\theta^t_{SSE} (\geq 0)$ implies $\mathbb{E}_a^{ossp} =  \mathbb{E}_a(\theta^t_{ossp}) < \mathbb{E}_a(\theta^t_{SSE}) = \mathbb{E}_a^{sse}$ (a contradiction) and $\theta^t_{ossp} < \theta^t_{SSE}$ implies $\mathbb{E}_a^{ossp} \ge  \mathbb{E}_a(\theta^t_{ossp}) > \mathbb{E}_a(\theta^t_{SSE}) = \mathbb{E}_a^{sse}$ (again, a contradiction). 
As a result, it must be the case that $\theta^t_{SSE} = \theta^t_{ossp}$ for all $\tau$, and thus $t$, as desired.

We now show that $ \mathbb{E}_a^{sse} = \mathbb{E}_a^{ossp}$ must hold true. Assume, for the sake of contradiction, that $ \mathbb{E}_a^{sse} > \mathbb{E}_a^{ossp}$. Then for any $\theta^t_{SSE} >0$, it must be that $\theta^t_{ossp} > \theta^t_{SSE}$. This is because $\theta^t_{ossp} \leq \theta^t_{SSE}$ implies that $\mathbb{E}_a^{ossp} \geq \mathbb{E}_a(\theta^t_{ossp}) \geq \mathbb{E}_a(\theta^t_{SSE}) = \mathbb{E}_a^{sse}$, which is a contradiction. On the other hand, for any $\theta^t_{ossp} >0$, $\theta^t_{SAG} > \theta^t_{SSE}$ must be true, because $0 < \theta^t_{ossp} \leq \theta^t_{SSE}$ implies that $\mathbb{E}_a^{sse} = \mathbb{E}_a(\theta^t_{SSE}) \leq \mathbb{E}_a(\theta^t_{ossp}) = \mathbb{E}_a^{ossp}$, which is a contradiction. As  a result, it must be the case that either $\theta^t_{SSE} = \theta^t_{ossp} = 0$ or $\theta^t_{ossp} > \theta^t_{SSE}$ for any $\tau$, thus $t$. Yet this contradicts the fact that $\sum_{\tau} \theta^t_{SSE} = \sum_{\tau} \theta^t_{ossp} = B_{\tau}$. Similarly, $\mathbb{E}_a^{sse} < \mathbb{E}_a^{ossp}$ can not hold true. As a result, $ \mathbb{E}_a^{sse} = \mathbb{E}_a^{ossp}$ is true. \qedhere
\end{proof}

In the proof above, we can conclude that the attacker's utility is the same in the OSSP and the online SSE. We now prove that the signaling procedure never weakens the auditor's expected utility.

\begin{theorem}\label{SAG_no_worse}
Given any game state,  the expected utility of the auditor by applying the OSSP is never worse than when the online SSE is applied.
\end{theorem}
\begin{proof}
If the attacker completes the attack, his expected utility by attacking type $t$ in SAG is $\mathbb{E}_a(\theta^t) = (p^{t}_1 + p^{t}_0) \cdot  U^{t}_{a,c}  + (q^{t}_1 + q^{t}_0) \cdot U^{t}_{a,u}$, where $\theta^{t}$ is the coverage probability of type $t$.
\begin{itemize}
\item If $\mathbb{E}_a(\theta^t) < 0 $, then the attacker will choose to not approach any target at the beginning, regardless of if there exists a signaling procedure. Thus, in both cases the auditor will achieve the same expected utility, which is $0$.
\item If $\mathbb{E}_a(\theta^t) \ge 0 $, then let $p^{t}_1 = 0$ and $q^{t}_1 = 0$. And it follows that $p^{t}_0 = \theta^{t}$ and $q^{t}_0 = 1-\theta^{t}$. This solution satisfies all of the constraints in LP \eqref{basic_game}, which, in this case, share exactly the same form with LP \eqref{SSG_SSE}. In combination with Theorem \ref{sig_zero}, we can conclude that in this special setting, the expected utilities of the auditor, by applying SAG (not necessary the OSSP) and online SSE, are the same: $\mathbb{E}_d(\theta^t ) =  \theta^{t} \cdot  U^{t}_{d,c}  + (1-\theta^{t})\cdot U^{t}_{d,u}$.  Thus, the expected utility of the auditor in the OSSP is never worse than the one in the online SSE.\qedhere
\end{itemize}
\end{proof}

This begs the following question: can applying the OSSP bring more benefit to the expected utility of the auditor? Our experiments lend support to an affirmative answer.

Our next result reveals an interesting property about the optimal signaling scheme. Interestingly, it turns out that by applying OSSP in specific situations, if there is no warning sent, then the auditor will not audit the triggered alerts in their optimal strategy (i.e., $p^{t_*}_0 = 0$).

\begin{theorem}\label{prob_0}
In SAG, if the payoff structure satisfies $0 \ge (U^{t_*}_{d,c} - P^{t_*} \cdot E^{t_*}_{\tau} \cdot C_{t_*})/(U^{t_*}_{d,u} - P^{t_*} \cdot E^{t_*}_{\tau} \cdot C_{t_*}) \ge U^{t_*}_{a,c}/U^{t_*}_{a,u}$ on the best attacking type $t_*$ in the OSSP, then we have $p_0^{t_*} = 0$ on the $\tau$--th alert.
\end{theorem}

\begin{proof}
This will be proved in the instance of LP \eqref{basic_game} that derives the best pair of the signaling strategy and the attacking strategy $t_*$. For inference convenience, for all $t$ we substitute $p^t_1$ and $q^{t}_1$ with $\theta^{t}_{ossp} - p^{t}_0$ and $1-\theta^{t}_{ossp} - q^{t}_0$, respectively. Combining with Theorem \ref{sig_zero}, the objective function of LP \eqref{basic_game} can be simplified as $p^{t_*}_0 \cdot U^{t_*}_{d,c} + q^{t_*}_0 \cdot  U^{t_*}_{d,u} + p^{t_*}_1 \cdot P^{t_*} \cdot E^{t_*}_{\tau} \cdot C_{t_*} +  q^{t_*}_1 \cdot P^{t_*} \cdot E^{t_*}_{\tau} \cdot C_{t_*} = p^{t_*}_0 \cdot (U^{t_*}_{d,c} - P^{t_*} \cdot E^{t_*}_{\tau} \cdot C_{t_*}) + q^{t_*}_0 \cdot  (U^{t_*}_{d,u} - P^{t_*} \cdot E^{t_*}_{\tau} \cdot C_{t_*}) + P^{t_*} \cdot E^{t_*}_{\tau} \cdot C_{t_*}$. 

Now, we simplify constraints. The first constraint is always tight in the OSSP (as shown in Theorem \ref{SSE_peSSE_equivalence}).
By applying the substitution rules, the second constraint becomes $\forall t', ~p^{t'}_0 \cdot U^{t'}_{a,c} + q^{t'}_0 \cdot U^{t'}_{a,u} \ge \theta^{t'}_{ossp} \cdot U^{t'}_{a,c} + (1-\theta^{t'}_{ossp}) \cdot U^{t'}_{a,u}$. For all $t' \neq t_*$, it can be future transformed into $(\theta^{t'}_{ossp} - p^{t'}_0 ) \cdot U^{t'}_{a,c} + (1 - \theta^{t'}_{ossp} - q^{t'}_0) \cdot U^{t'}_{a,u} \le 0.$ Due to the fact that $p_1^{t'}=q_1^{t'}=0$ in the OSSP for $\forall t' \neq t_{*}$, $p^{t'}_0$ is equal to $\theta^{t'}_{ossp}$, and $q^{t'}_0$ equal to $1-\theta^{t'}_{ossp}$. As such, for all $t' \neq t_*$, this constraint naturally holds true. By far, the best strategy pair of SAG in our setting needs to maximize $p^{t_*}_0 \cdot (U^{t_*}_{d,c} - P^{t_*} \cdot E^{t_*}_{\tau} \cdot C_{t_*}) + q^{t_*}_0 \cdot  (U^{t_*}_{d,u} - P^{t_*} \cdot E^{t_*}_{\tau} \cdot C_{t_*}) + P^{t_*} \cdot E^{t_*}_{\tau} \cdot C_{t_*}$, such that  $p^{t_*}_0 \cdot U^{t_*}_{a,c} + q^{t_*}_0 \cdot U^{t_*}_{a,u} \ge \theta^{t_*}_{ossp} \cdot U^{t_*}_{a,c} + (1-\theta^{t_*}_{ossp}) \cdot U^{t_*}_{a,u}$ (we refer to this inequality as constraint $\alpha$) and that these probability variables are in $[0,1]$ and sum up to $1$.\footnote{Constraints involving $B^{t'}_{\tau}$ are neglected because $\theta^{t'}_{ossp}$ is the coverage probability that can be derived from $B^{t'}_{\tau}$ in our setting.} 

We set up a Cartesian coordinate system and let $q^{t_*}_0$ be the vertical axis and $p^{t_*}_0$ the horizontal one. Geometrically, the slopes of the item to be maximized, which is $- (U^{t_*}_{d,c}-P^{t_*} \cdot E^{t_*}_{\tau} \cdot C_{t_*}) / (U^{t_*}_{d,u}-P^{t_*} \cdot E^{t_*}_{\tau} \cdot C_{t_*})$ and constraint $\alpha$, which is $- U^{t_*}_{a,c}/U^{t_*}_{a,u} $ are both positive.  Note that, though we do not constrain the left side of constraint $\alpha$, which is $\mathbb{E}^{t_*}_a(util  | \xi_u) = p^{t_*}_0 \cdot  U^{t_*}_{a,c}  + q^{t_*}_0 \cdot U^{t_*}_{a,u} > 0$, this inequality is always true. If not the case, the attacker will not initially attack. We discuss the righthand side $\beta = \theta^{t_*}_{ossp} \cdot U^{t_*}_{a,c} + (1-\theta^{t_*}_{ossp}) \cdot U^{t_*}_{a,u}$ as follows.
\begin{itemize}
\item $\beta \le 0$.   In this setting, constraint $\alpha$ is dominated. The boundary of the dominant constraint passes the origin and the feasible region is a triangle with its base on the vertical axis, as shown in Figure \ref{feas_a1}. Thus, in both cases, if $(U^{t_*}_{d,c} - P^{t_*} \cdot E^{t_*}_{\tau} \cdot C_{t_*})/(U^{t_*}_{d,u} - P^{t_*} \cdot E^{t_*}_{\tau} \cdot C_{t_*}) \ge U^{t_*}_{a,c}/U^{t_*}_{a,u}$ holds true (which implies that the slope of the objective function is less than the boundary's slope of the dominant constraint), then $p^{t_*}_0 = q^{t_*}_0 = 0$ leads to the maximum of the objective function. The OSSP, thus is  $p^{t_*}_1 = \theta^{t*}_{ossp} , q^{t_*}_1 = 1- \theta^{t_*}_{ossp}, p^{t_*}_0 = q^{t_*}_0 = 0$.
\item $\beta > 0$. Thus, constraint $\alpha$ dominates $p^{t_*}_0 \cdot  U^{t_*}_{a,c}  + q^{t_*}_0 \cdot U^{t_*}_{a,u} > 0$. The boundary's intercept of the dominant constraint is $\delta  = (\theta^{t_*}_{ossp} \cdot U^{t_*}_{a,c} + (1-\theta^{t_*}_{ossp}) \cdot U^{t_*}_{a,u})/U^{t_*}_{a,u} \in (0,1]$. Using an analysis similar to the previous case of $\beta$, only when $p^{t_*}_0 = 0, q^{t_*}_0 = \delta$ does lead to the maximum of the objective function. This is indicated in Figure \ref{feas_a2}. The OSSP is $p^{t_*}_1 = \theta^{t_*}_{ossp}, p^{t_*}_0 = 0, q^{t_*}_1 = 1-\theta^{t_*}_{ossp} - (\theta^{t_*}_{ossp} \cdot U^{t_*}_{a,c} + (1-\theta^{t_*}_{ossp}) \cdot U^{t_*}_{a,u})/U^{t_*}_{a,u}, q^{t_*}_0 = (\theta^{t_*}_{ossp} \cdot U^{t_*}_{a,c} + (1-\theta^{t_*}_{ossp}) \cdot U^{t_*}_{a,u})/U^{t_*}_{a,u}.$ \qedhere
\end{itemize}
\end{proof}

\begin{figure}[h]%
\vspace{-4mm}
\centering
\label{feas_a}
\subfigure[$\beta \le 0$]{%
\label{feas_a1}%
\includegraphics[width=4.5cm]{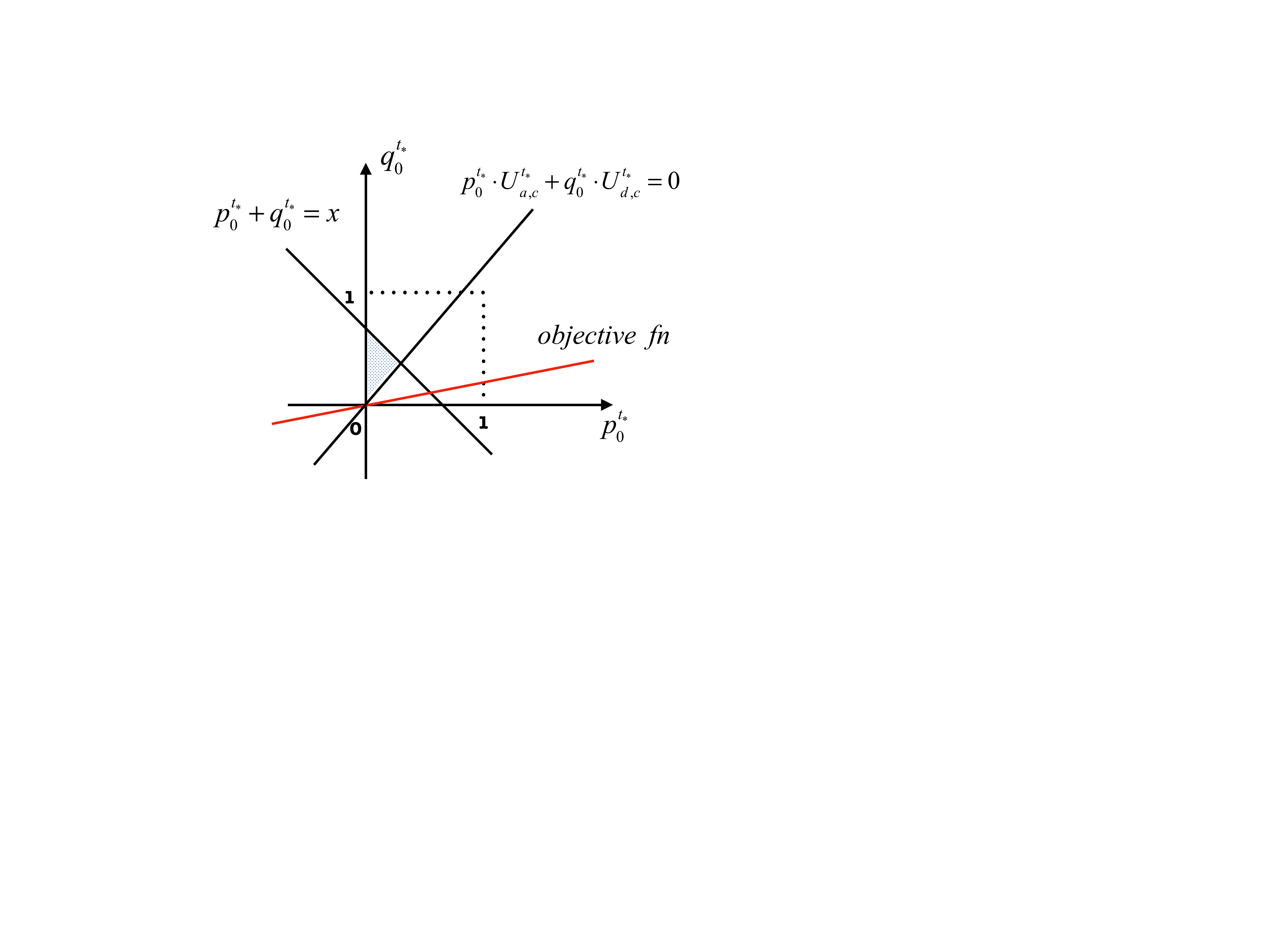}}%
\subfigure[$\beta > 0$]{%
\label{feas_a2}%
\includegraphics[width=4.5cm]{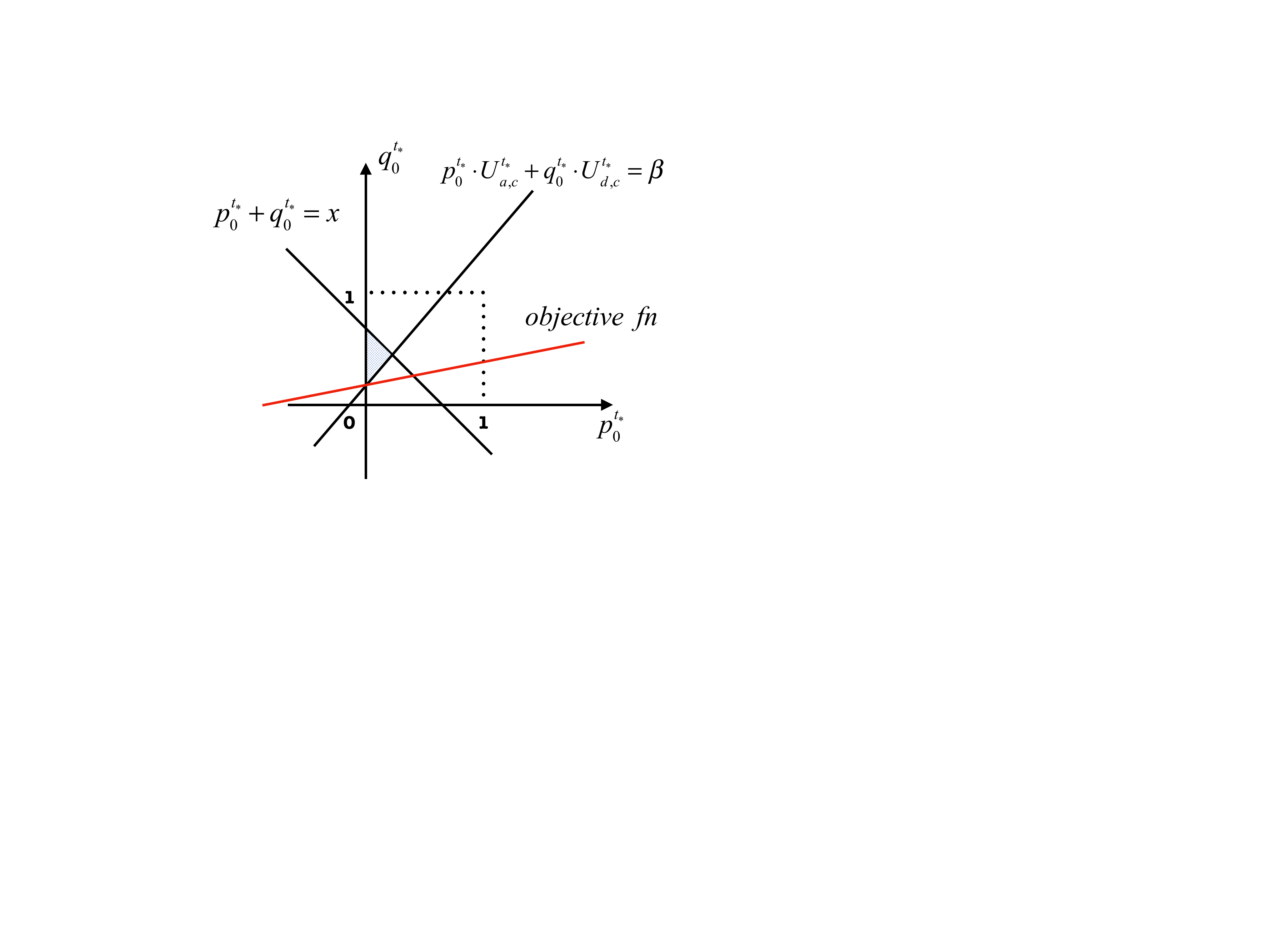}}%
\caption{Feasible regions (blue areas) and an objective function gaining the largest value for $\beta \le 0$ and $\beta>0$. Note that the boundary $p_0^{t_*} + q_0^{t_*} = x$ is only for illustration, and its intercept can slide in $[0,1]$ by taking into account the value of $p_1^{t_*}$ and $p_1^{t_*}$. However, this never impact the optimal solution.}
\end{figure}

\noindent\textbf{Remark. } In application domains, the absolute value of the penalty for the attacker is often greater than the benefit from committing attacks. As for the auditor, his/her benefit from catching an attack is often less than the absolute value of the loss due to missing an attack. If the warning cost $P^{t_*} \cdot E^{t_*}_{\tau} \cdot C_{t_*}$ were ignored, then $0 \ge U^{t_*}_{d,c}/U^{t_*}_{d,u} \ge U^{t_*}_{a,c}/U^{t_*}_{a,u}$ is often satisfied in practice. Considering that the warning cost is proportional to the estimation of the number of future warning events, which decreases with time, the condition in Theorem \ref{prob_0} only happens in a certain period of time.

One might wonder that, given that the condition in Theorem \ref{prob_0} is valid, whether the attacker can keep attacking until receiving no warning, in which case the attacker can attack safely under the optimal signaling scheme? Actually, this strategy cannot lead to success because once the attacker chooses to quit, his/her identity is essentially revealed. The auditor cannot punish the attacker (yet) because the attacker quits the attack, leaving no evidence. Therefore, a successful attack later on only hurts him/her, while help the auditor find forensic evidence of an attack. In practice, it is common that the auditor uses reserved budget to deal with special cases. In the setting above, the author can use a small portion of the auditing budget to investigate repeated attempts of data access, but in practice this is not an issue, as these cases are likely to be rare in real world.
As a result, once an attacker chooses to quit, the best response should be to not attack during the rest of the auditing cycle. In the experimental comparison with online/offline SSG, which requires no additional budget for such attack category, we will apply a reduced available budget as the input of the corresponding SAG to ensure fairness in our comparisons.

A natural follow-up question is can the attacker manipulate the model by running this strategy across audit cycles? The answer is no as well. Such a behavior can be easily detected by a rule that applies when the attacker performs his/her attack repeatedly. When the auditor does not send a warning, the attacker successfully attacks. Yet, since there was a warning sent previously, the auditor will use the probability $p_1$ to audit, rather than $p_0$. Thus, the attacker should take this into account before adopting such a strategy. 


\begin{theorem}\label{SSE_OSSP_Equal}
The auditor benefits equally in terms of the expected utility from SAG and online SSG at the $\tau$-th alert, if it satisfies $U^{t_*}_{d,u} >  P^{t_*} \cdot E^{t_*}_{\tau} \cdot C_{t_*}$, where $t_*$ is the best type to attack in the OSSP.
\end{theorem}
\begin{proof}
We prove this by applying the same simplification and the split strategy (i.e., analyze two distinct situations based on the value of $\beta$) as applied in the proof for Theorem \ref{prob_0}. Note that the slope of the objective function is $-(U^{t_*}_{d,c} - P^{t_*} \cdot E^{t_*}_{\tau} \cdot C_{t_*})/(U^{t_*}_{d,u} - P^{t_*} \cdot E^{t_*}_{\tau} \cdot C_{t_*})$. Since $C_{t_*} < 0$, the numerator is less than $0$. If $U^{t_*}_{d,u} >  P^{t_*} \cdot E^{t_*}_{\tau} \cdot C_{t_*}$, then the denominator is greater than 0. Thus, the slope is less than $0$. In particular, the slope is less than $-1$ (which is the slope of boundary $p_0^{t_*} + q_0^{t_*} = x$) because of $U^{t_*}_{d,c}\ge0>U^{t_*}_{d,u}$. We now analyze properties in this situation geometrically.

As demonstrated in Figures \ref{feas_b1} and \ref{feas_b2}, the boundary $p_0^{t_*} + q_0^{t_*} = x (\in [0,1])$ should pass through the $(0,1)$ point. This is because, if this failed to occur, then the value of the objective function can be further improved by lifting the boundary. The optimal solution for both cases is at the intersection point of the two boundaries of the feasible region. Thus,  it follows that $p_0^{t_*} + q_0^{t_*} = 1$ for the OSSP, which implies  $p_1^{t_*} = q_1^{t_*} = 0$. In other words, the signaling procedure is turned off for the best attacking type $t_*$ in the OSSP. Combining with what Theorem \ref{sig_zero} indicates, when $U^{t_*}_{d,u} >  P^{t_*} \cdot E^{t_*}_{\tau} \cdot C_{t_*}$, the signaling procedure is off for all types. In LP \eqref{basic_game}, by substituting variables $p^{t'}_1$ and $q^{t'}_1$ (for all $t'$) with $0$, the SAG instance becomes an online SSG (as shown in LP \eqref{SSG_SSE}). Thus, the two LPs share the same solution, and the auditor will receive the same expected utility in both auditing mechanism.\qedhere

\end{proof}

\begin{figure}[h]%
\vspace{-4mm}
\centering
\label{feas_b}
\subfigure[$\beta \le 0$]{%
\label{feas_b1}%
\includegraphics[width=4.6cm]{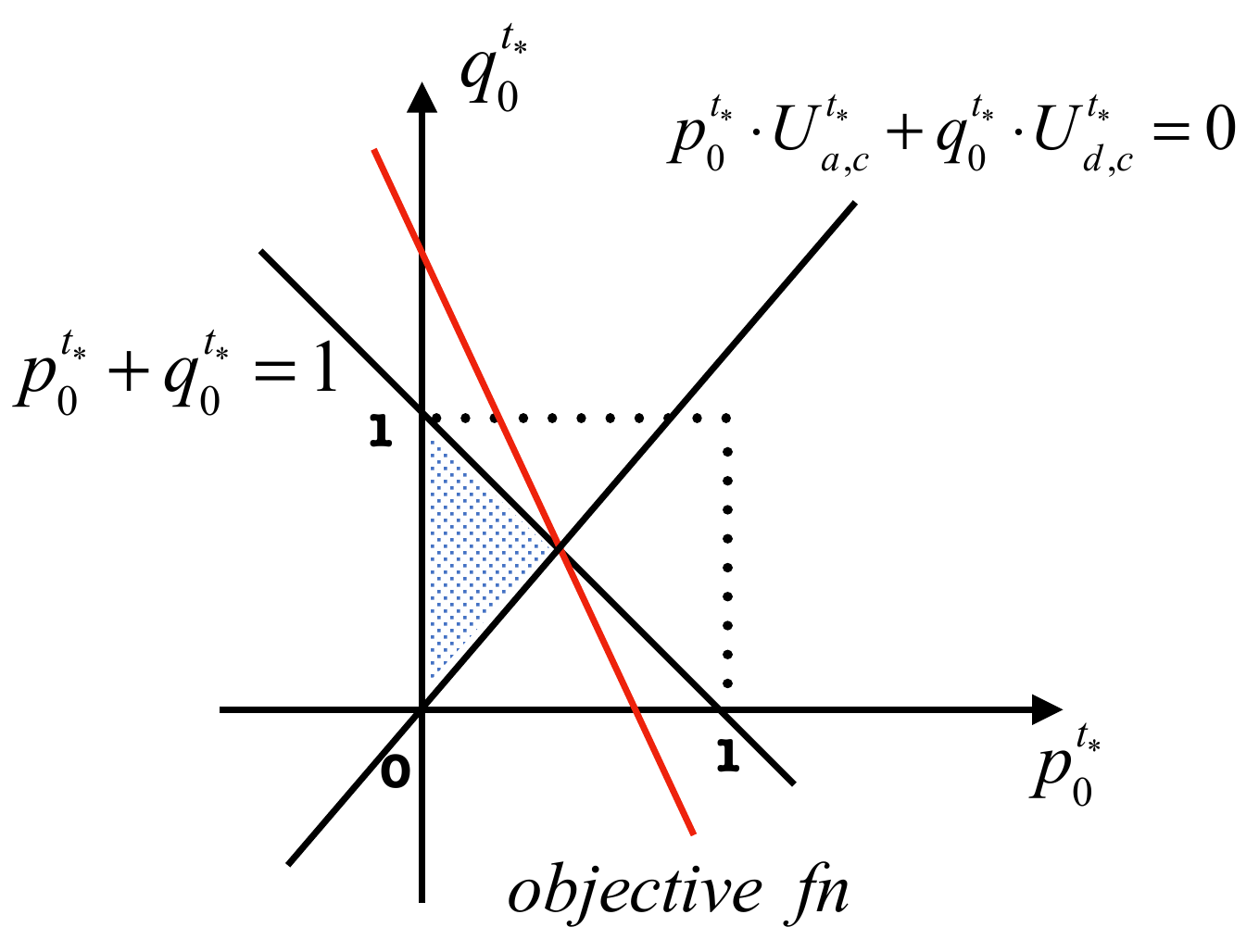}}%
\subfigure[$\beta > 0$]{%
\label{feas_b2}%
\includegraphics[width=4.4cm]{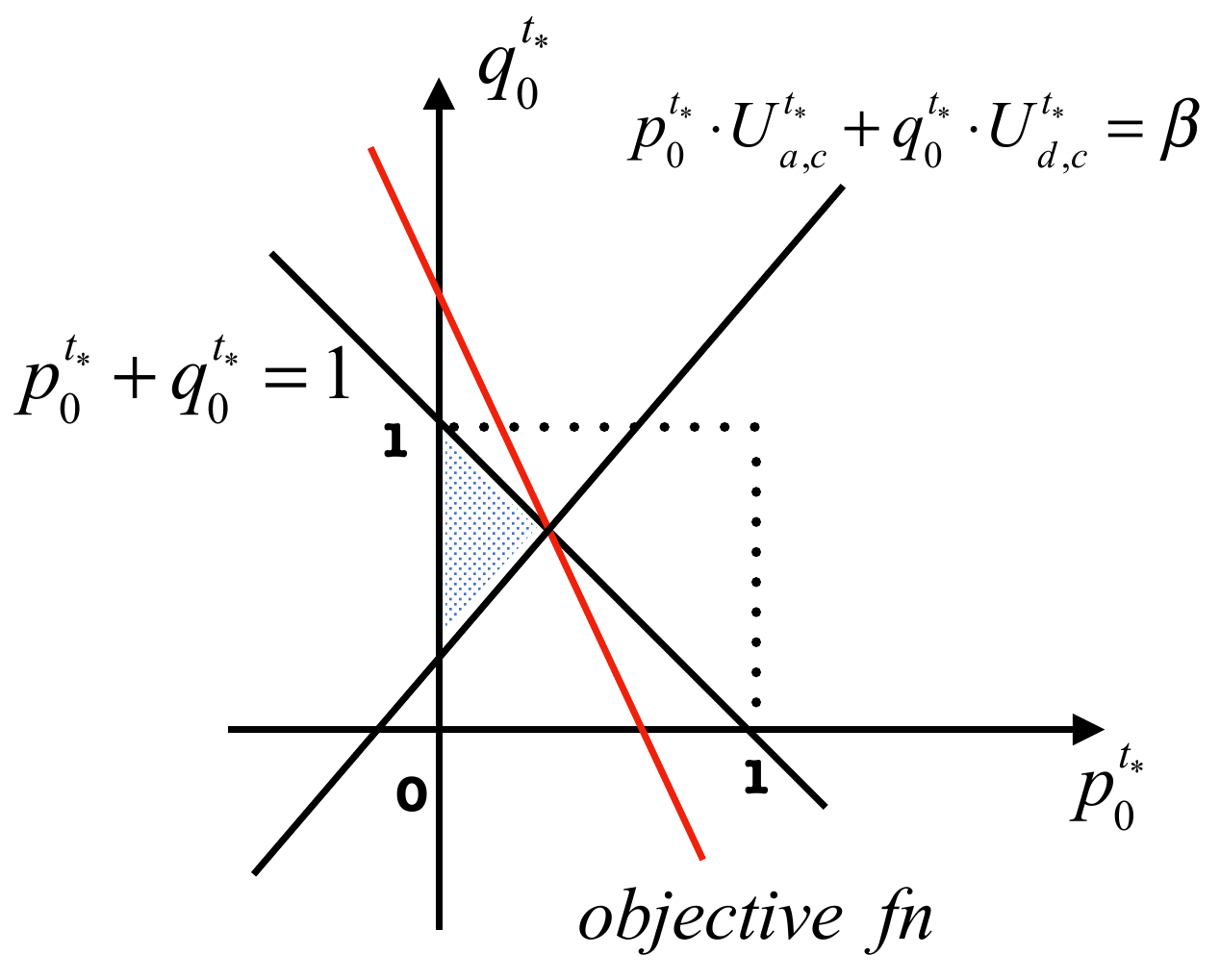}}%
\caption{Feasible regions (blue shaded triangle areas) and an objective function gaining the largest value for $\beta \le 0$ and $\beta>0$.  }
\end{figure}

This result indicates that if the incurred loss due to a warning is too large, then an SAG will degrade into an online SSG, where the signaling procedure is turned off. It suggests that in the application domain, to ensure that the signaling is deployed in a useful manner, organizations need to 1) refine the alert system so that false positive alerts can be classified as normal events, and 2) decrease  the number of the events in which normal users are scared away.

%% file: VI_evaluation.tex
\section{Model Evaluation}
In this section, we evaluate the performance of the SAG on the real EMR access logs from VUMC, which deployed an unoptimized warning strategy. To illustrate the value of signaling, we compare with multiple game theoretic alternative methods in terms of the expected utility of the auditor. Specifically, we investigate the robustness of the advantage of SAGs under a range of different conditions. Now we first describe the real dataset which is used for evaluation. 

\begin{table}
	\fontsize{7.5}{8}\selectfont
	\centering
	\caption{A summary of the daily statistics per alert types.}\label{type}
	\begin{tabular}{llcccccc}
		\toprule
		\textbf{ID}  & \textbf{Alert Type Description}  & \textbf{Mean}  & \textbf{Std}  \\[0.1em]\hline
		\\[-0.2em]
		1 & Same Last Name & $196.57$ & $17.30$   \\[0.2em]
		2 & Department Co-worker &  $29.02$ & $5.56 $  \\[0.2em]
		3 & Neighbor ($\le$ 0.5 miles) &  $140.46$ & $23.23$   \\[0.2em]
		4 & Same Address & 10.84 &  $3.73$\\[0.2em]
		5 & Last Name; Neighbor ($\le$ 0.5 miles) & $25.43$  & $4.51$   \\[0.2em]
		6 & Last Name; Same Address & $15.14$  & $4.10$   \\[0.2em]
		7 & Last Name; Same Address; Neighbor ($\le$ 0.5 miles) & $43.27$  & $6.45$   \\[0.1em]
		\hline
	\end{tabular}
	\vspace{-2mm}
\end{table}

\begin{table}[h]
	\fontsize{7.1}{8}\selectfont
	\centering
	\caption{The payoff structures for the pre-defined alert types.}\label{payoff}
	\begin{tabular}{rrrrrrrr}
		\toprule
		\\[-1.0em]
		\textbf{Payoff}  &  \textbf{Type} $1$  & \textbf{Type} $ 2$  & \textbf{Type} $3$ & \textbf{Type} $4$ & \textbf{Type} $5$ & \textbf{Type} $6$ & \textbf{Type} $7$ \\[0.2em]\hline
		\\[-0.4em]
		 $U_{d,c}$  & $100$ & $150$ & $150$ & $300$ &  $400$ & $600$ & $700$ \\[0.2em]
		 $U_{d,u}$  & $-400$ & $-500$ & $-600$ & $-800$ & $-1000$ & $-1500$ & $-2000$ \\[0.2em]
		 $U_{a,c}$  & $-2000$ & $-2250$ & $-2500$ & $-2500$ & $-3000$ & $-5000$ & $-6000$ \\[0.2em]
		 $U_{a,u}$  & $400$ & $400$ & $450$ & $600$ & $650$ & $700$ & $800$ \\[0.1em]
		\hline
	\end{tabular}
	\vspace{-2mm}
\end{table}

\subsection{Dataset}
The dataset consists of EMR access logs for 56 continuous normal working days in 2017. We excluded all holidays (include weekends) because they exhibit a different access pattern from working days. The total number of unique accesses $\langle$Date, Employee, Patient$\rangle$ is on the order of $10.75M$. The mean and standard deviation of daily unique accesses are approximately $192K$ and $8.97K$, respectively. 
We focus on the following alerts types, which are real rules deployed for potential auditing: 1) employee and patient share the same last name, 2) employee and patient work in the same VUMC department, 3) employee and patient share the same residential address, and 4) employee and patient live within $0.5$ miles of one another. When an access triggers multiple distinct types of alerts, their combination is regarded as a new type. 
Table \ref{type} lists the set of predefined alert types, along with the mean and standard deviation of their occurrence on a daily basis. We provide the payoff structure for both the attacker and the auditor in Table \ref{payoff}, the magnitude of which is based on the input of a domain expert.

\subsection{Experimental Setup}

\begin{figure*}[h]%
\label{fig4}
\centering
\subfigure[ Day 1]{%
\label{fig_mul_1}%
\includegraphics[width=4.0cm]{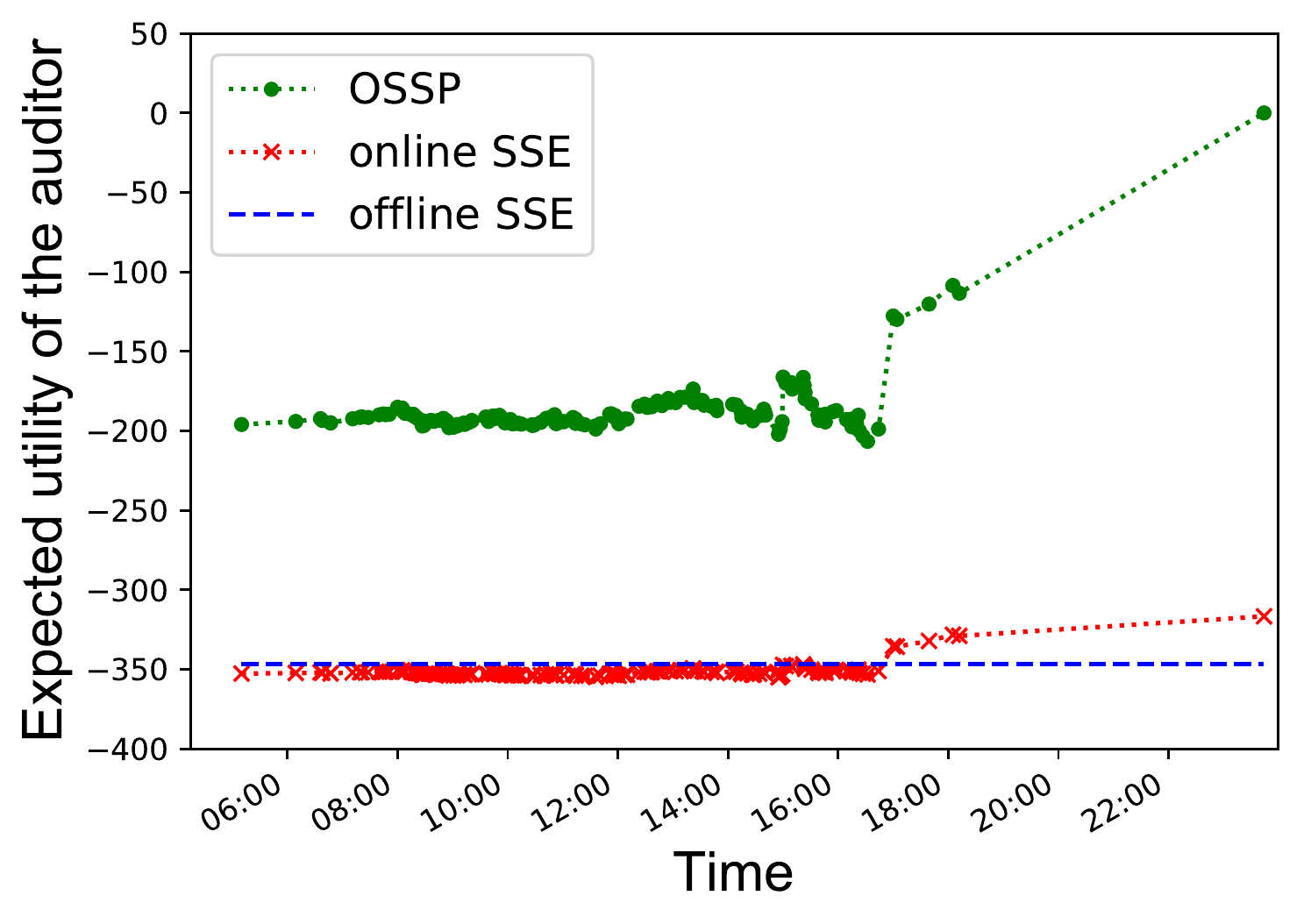}}%
\subfigure[ Day 2]{%
\label{fig_mul_2}%
\includegraphics[width=4.0cm]{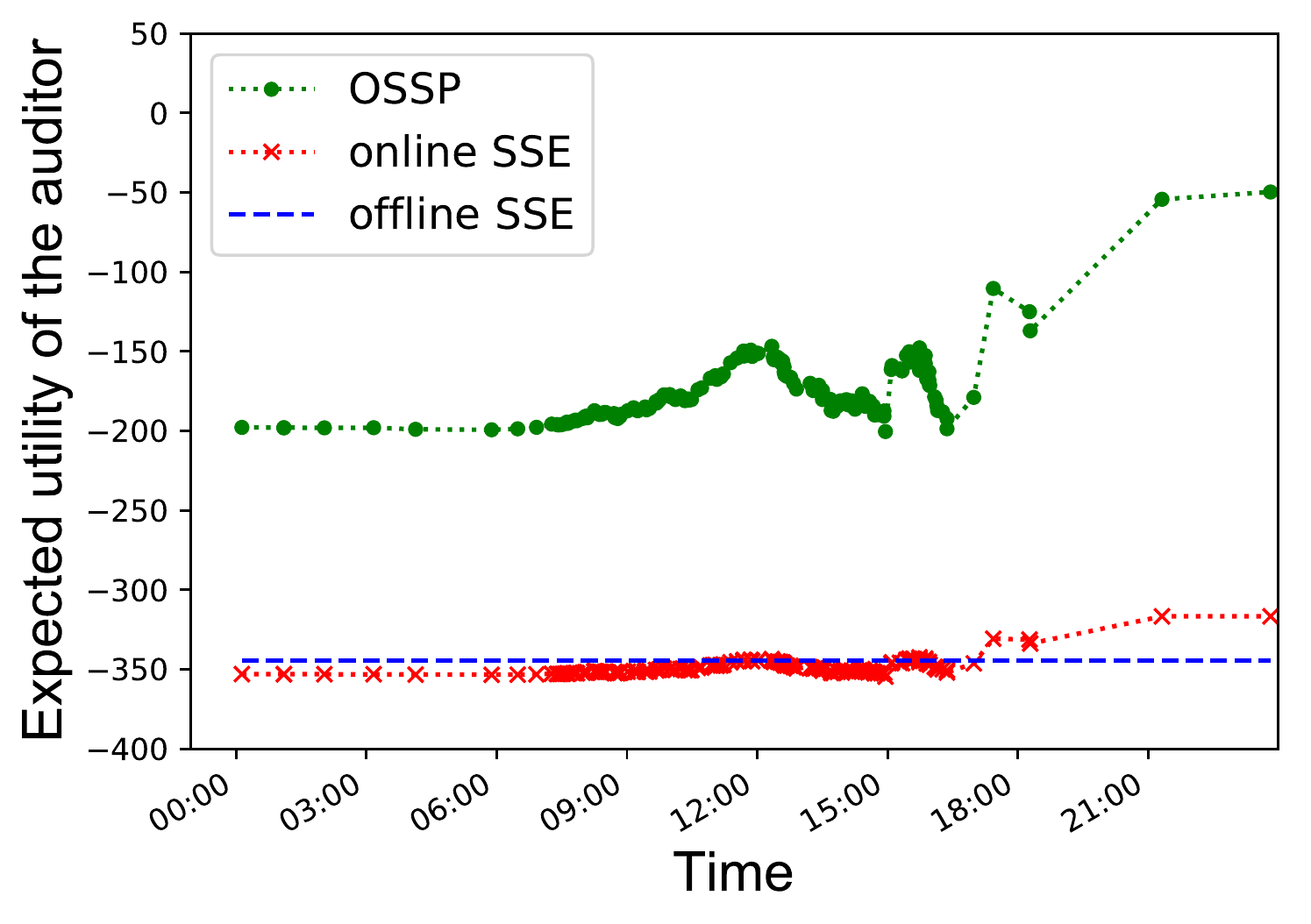}}%
\subfigure[ Day 3]{%
\label{fig_mul_3}%
\includegraphics[width=4.0cm]{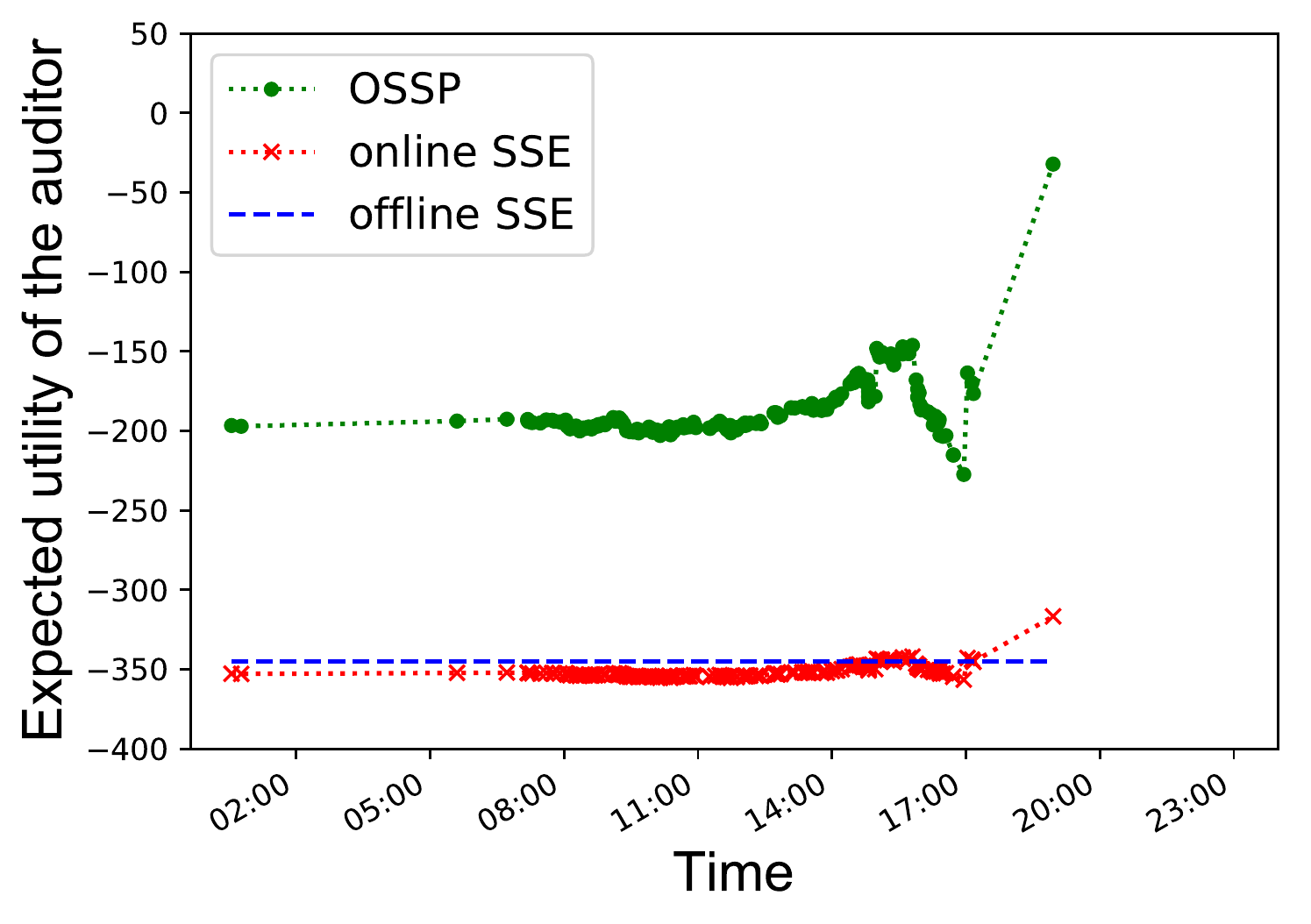}}%
\subfigure[ Day 4]{%
\label{fig_mul_4}%
\includegraphics[width=4.0cm]{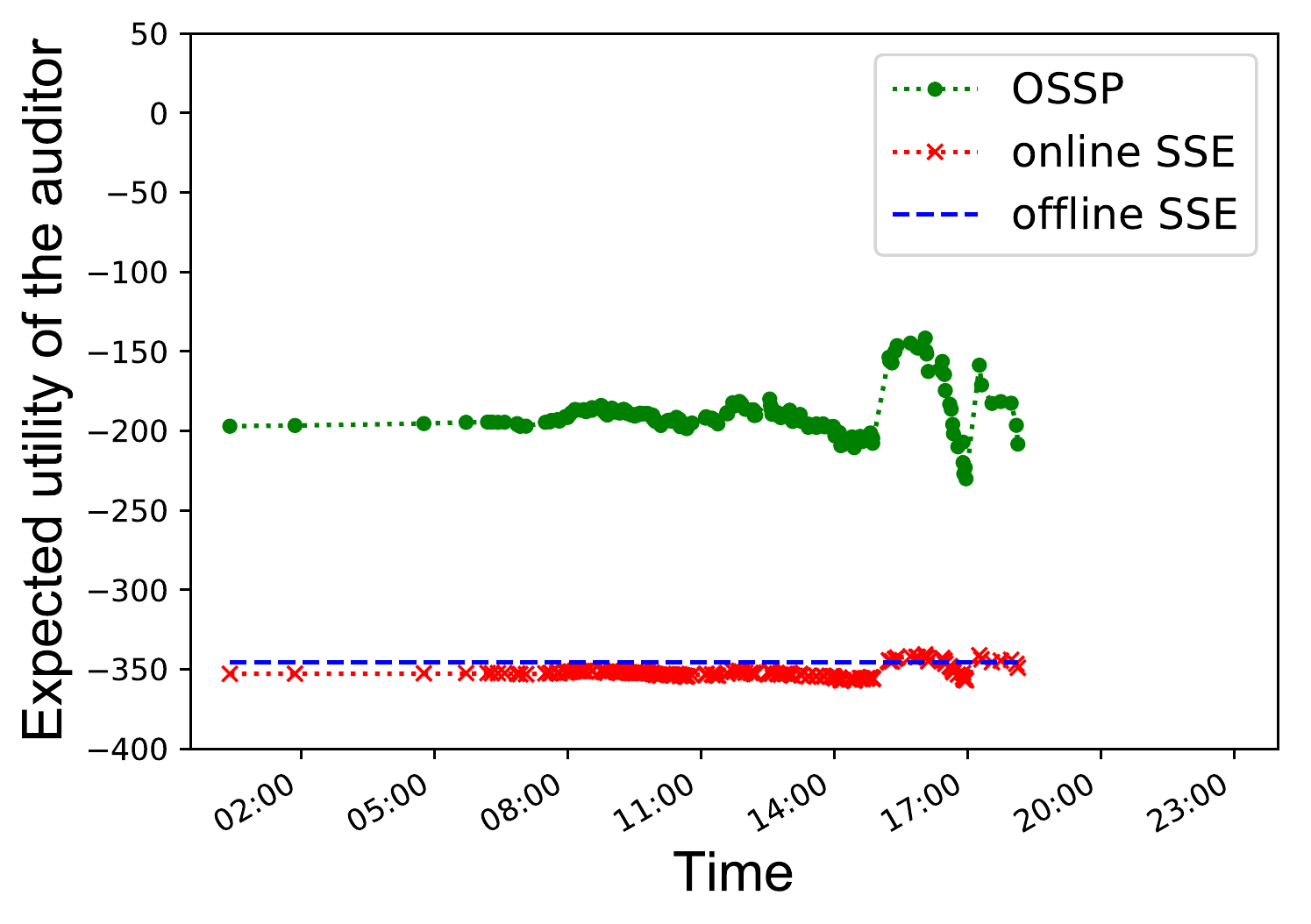}}%

\subfigure[ Day 5]{%
\label{fig_mul_5}%
\includegraphics[width=4.0cm]{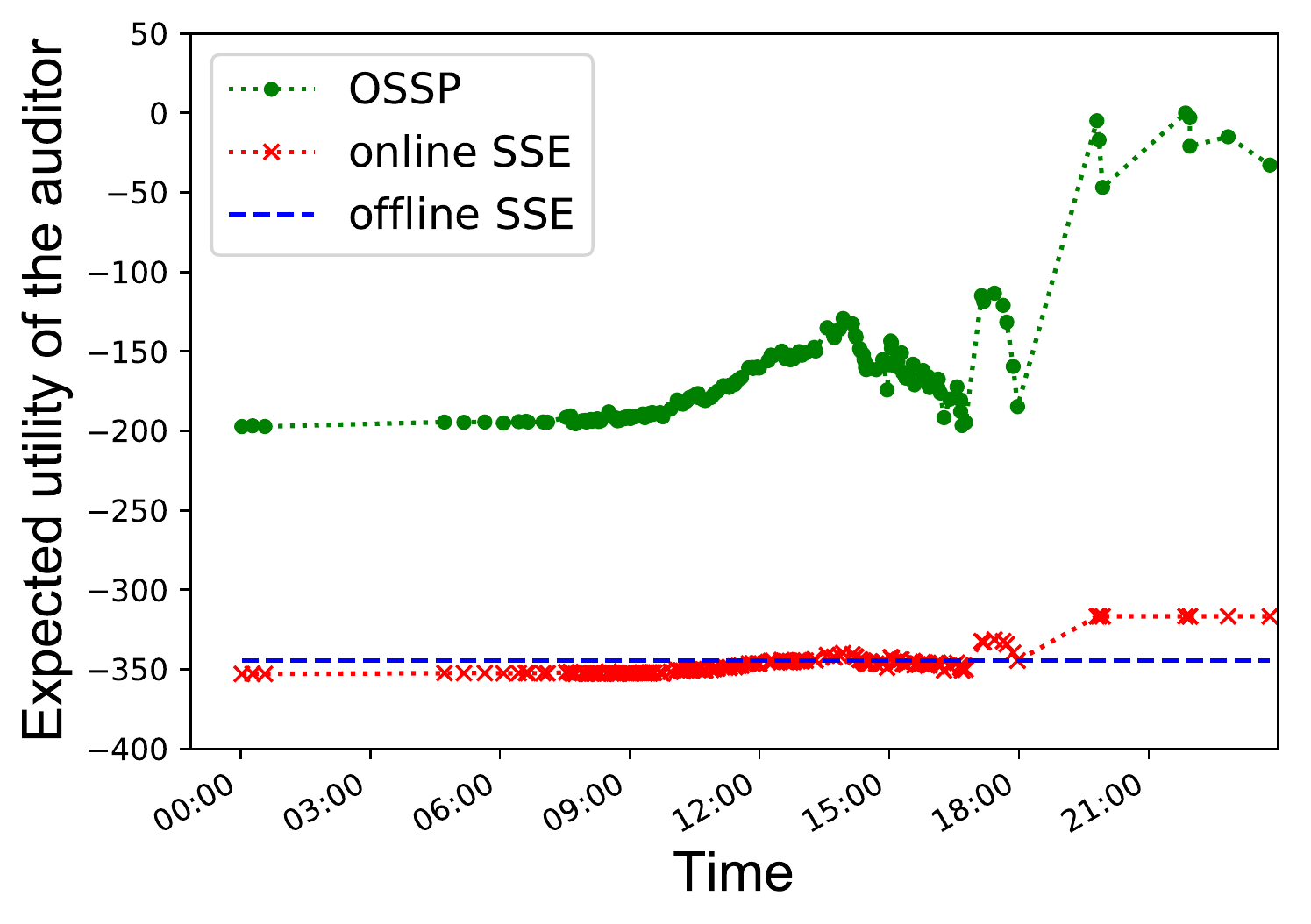}}%
\subfigure[ Day 6]{%
\label{fig_mul_6}%
\includegraphics[width=4.0cm]{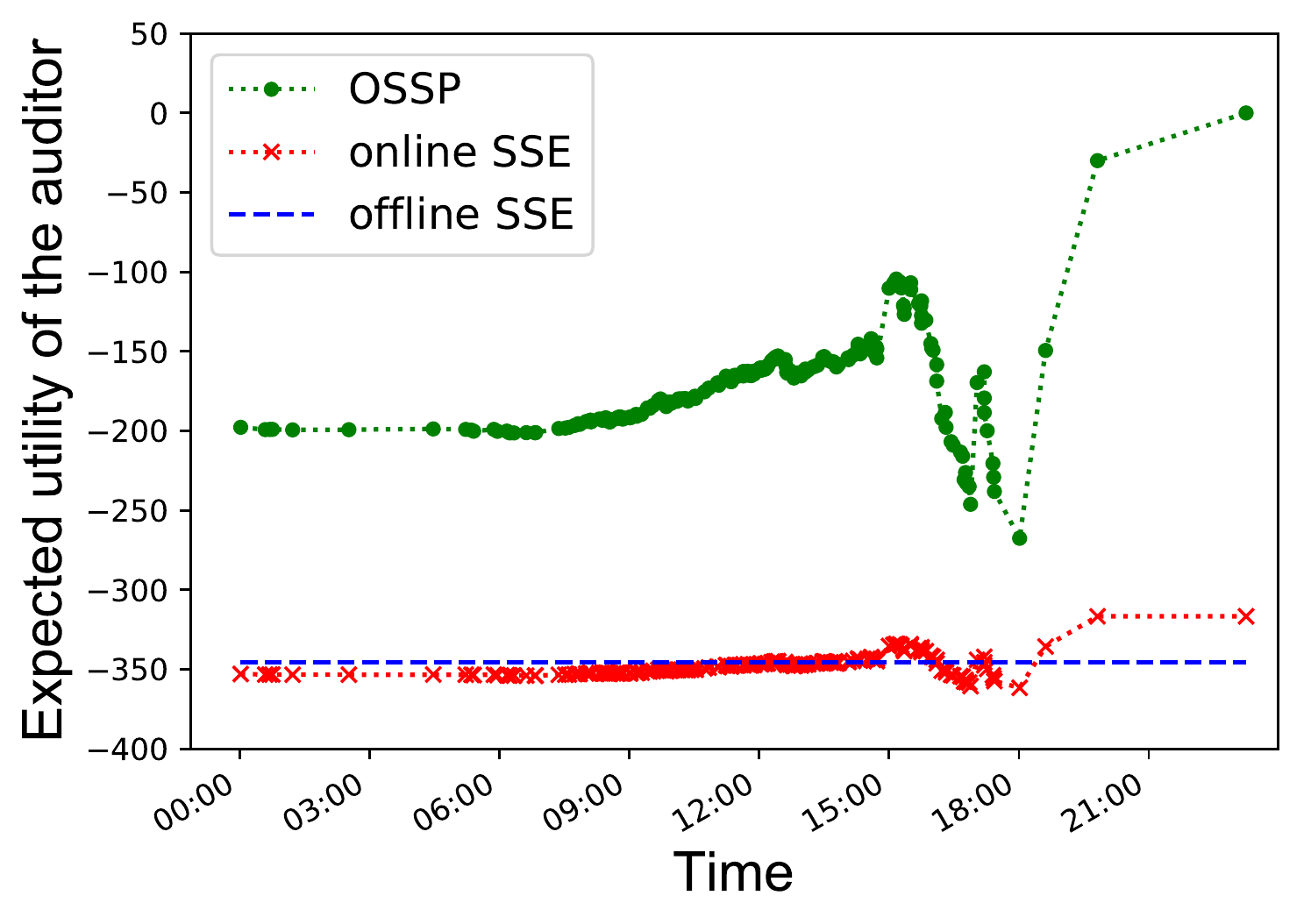}}%
\subfigure[ Day 7]{%
\label{fig_mul_7}%
\includegraphics[width=4.0cm]{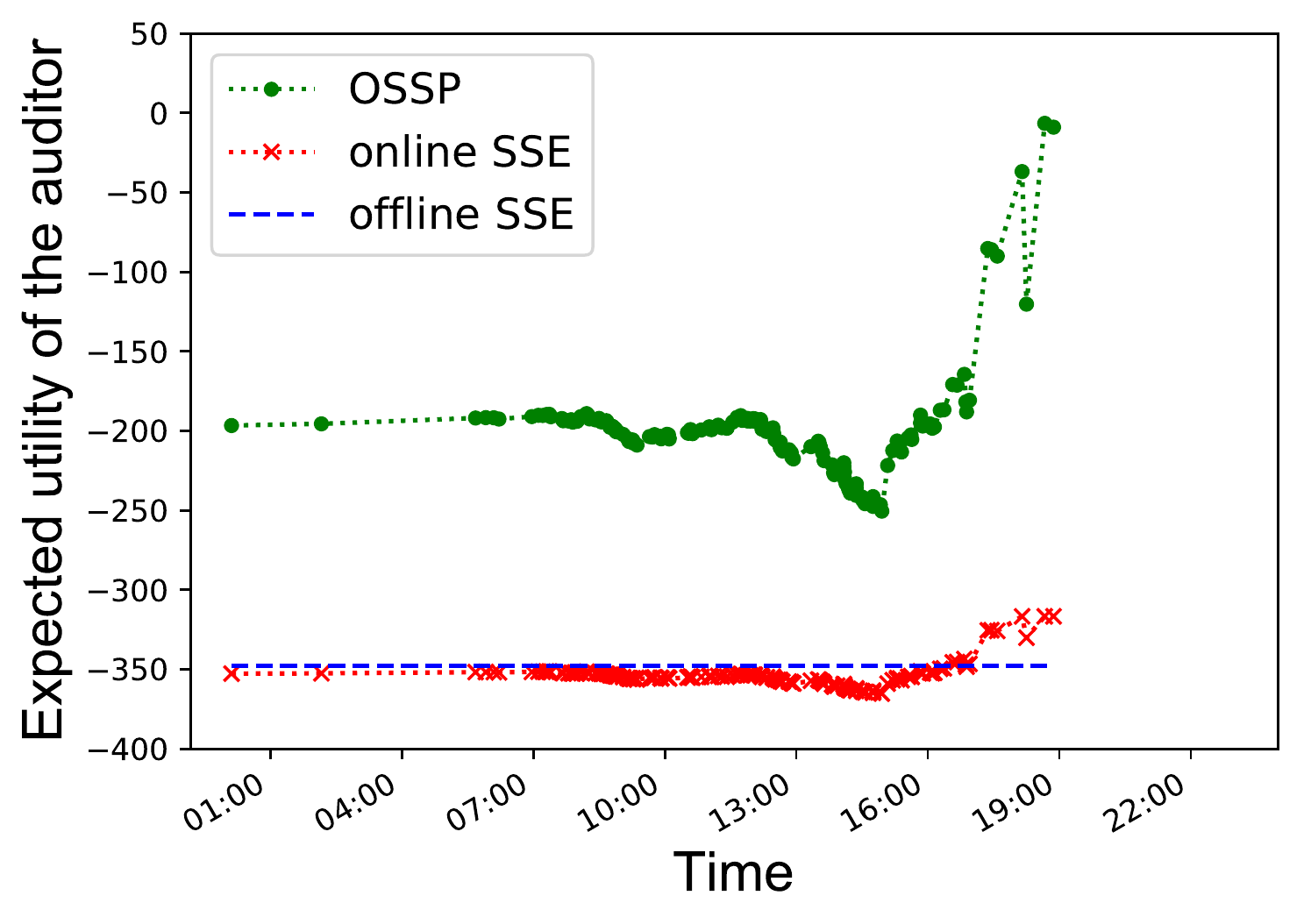}}%
\subfigure[ Day 8]{%
\label{fig_mul_8}%
\includegraphics[width=4.0cm]{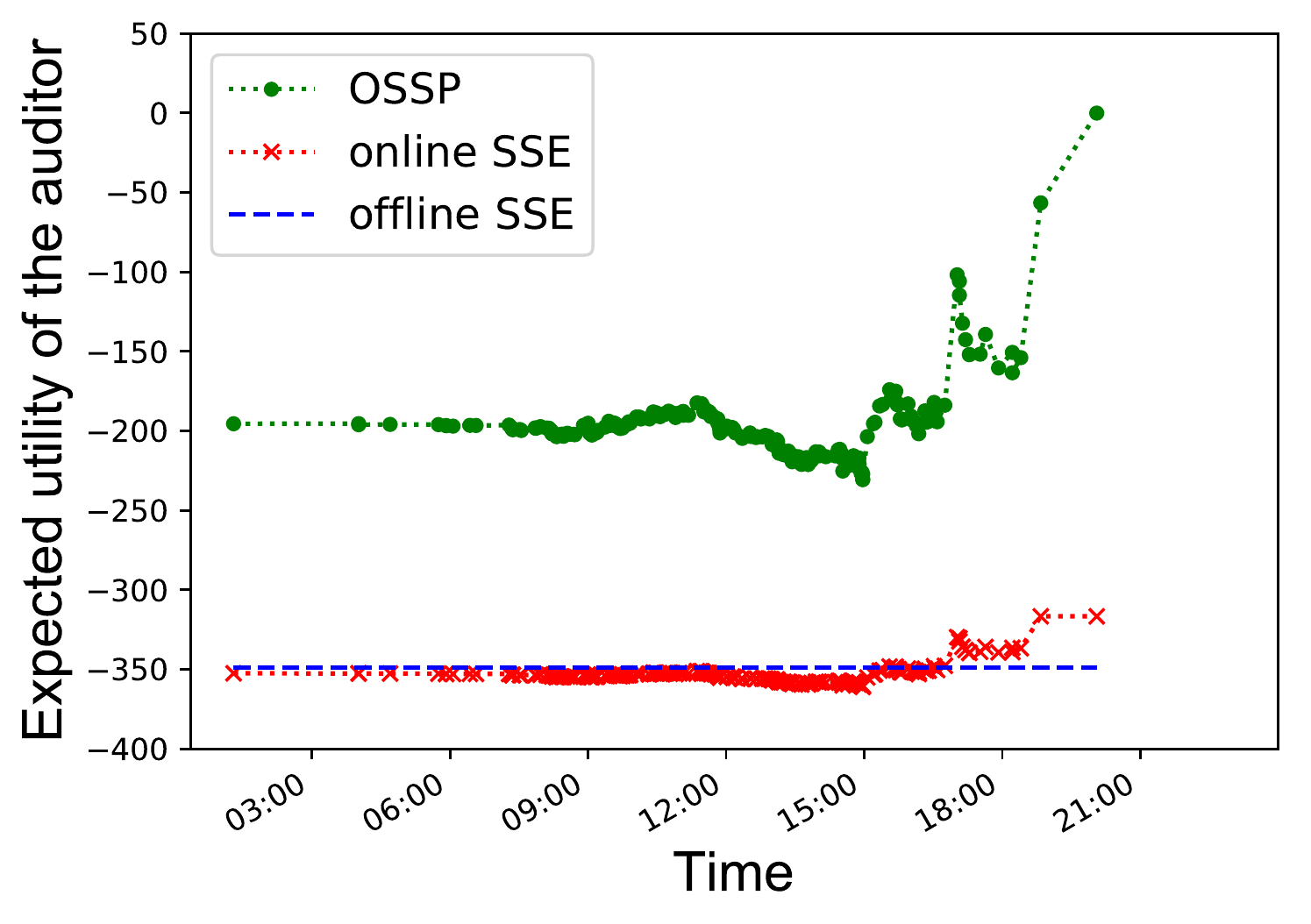}}%

\subfigure[ Day 9]{%
\label{fig_mul_9}%
\includegraphics[width=4.0cm]{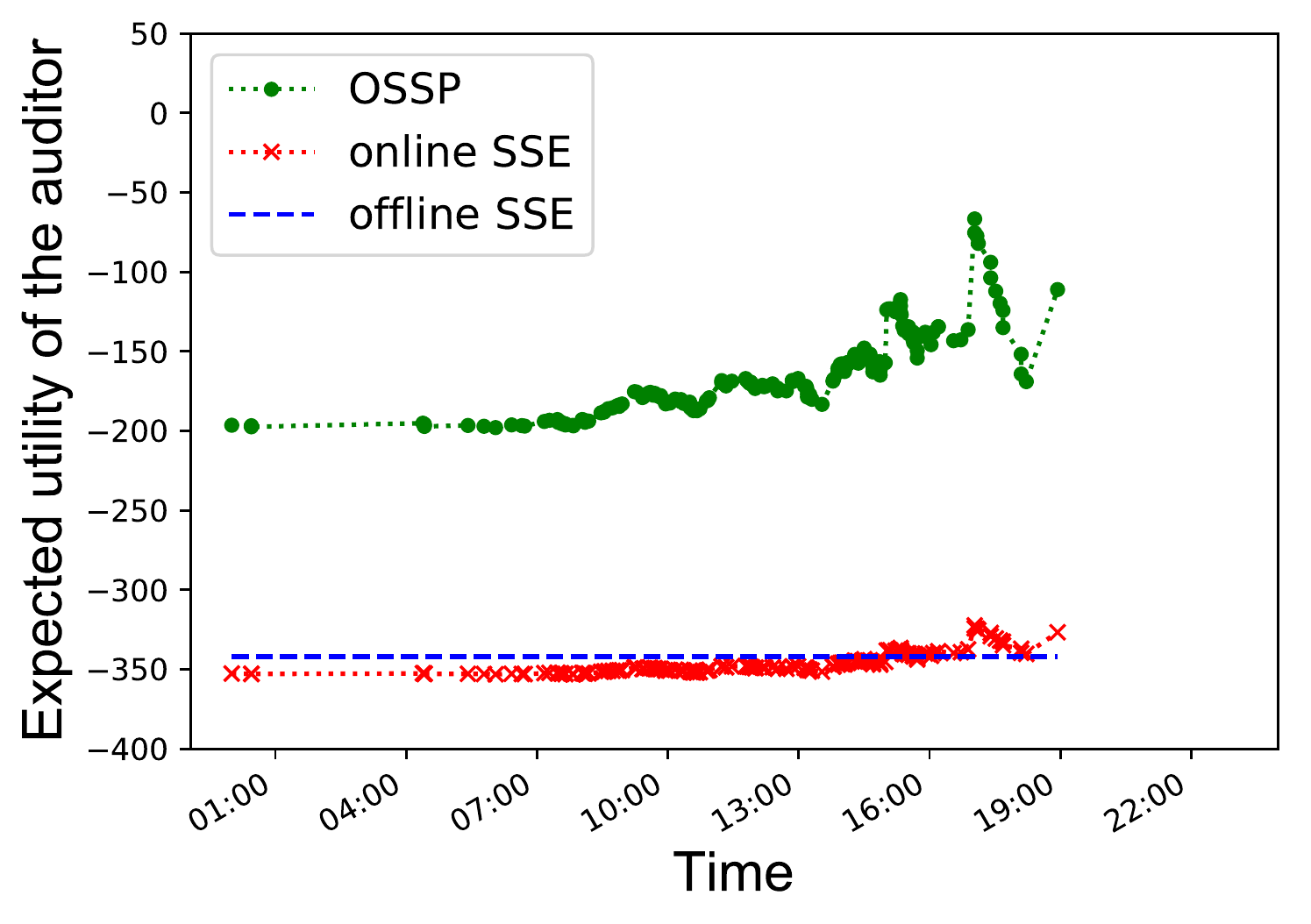}}%
\subfigure[ Day 10]{%
\label{fig_mul_10}%
\includegraphics[width=4.0cm]{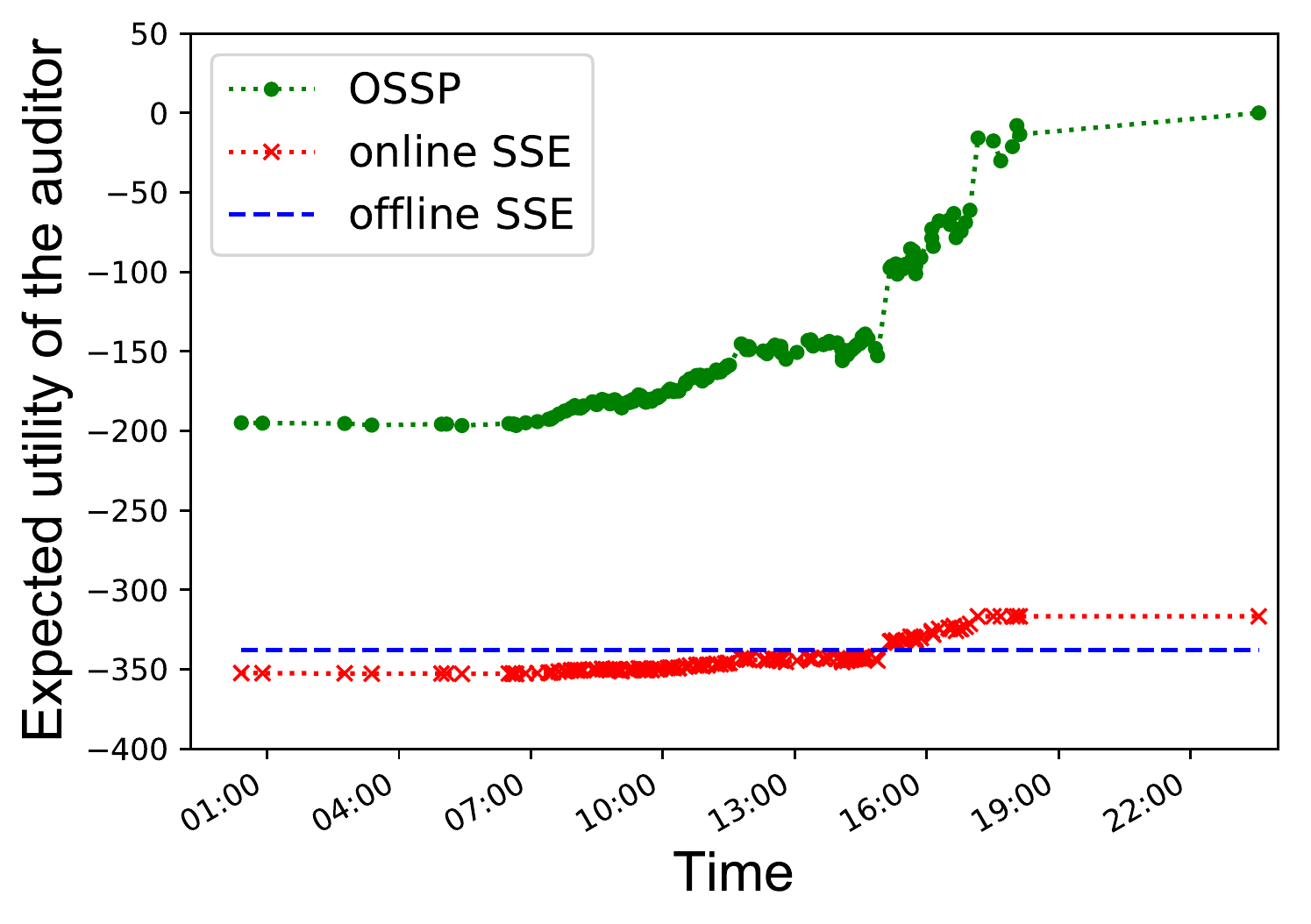}}%
\subfigure[ Day 11]{%
\label{fig_mul_11}%
\includegraphics[width=4.0cm]{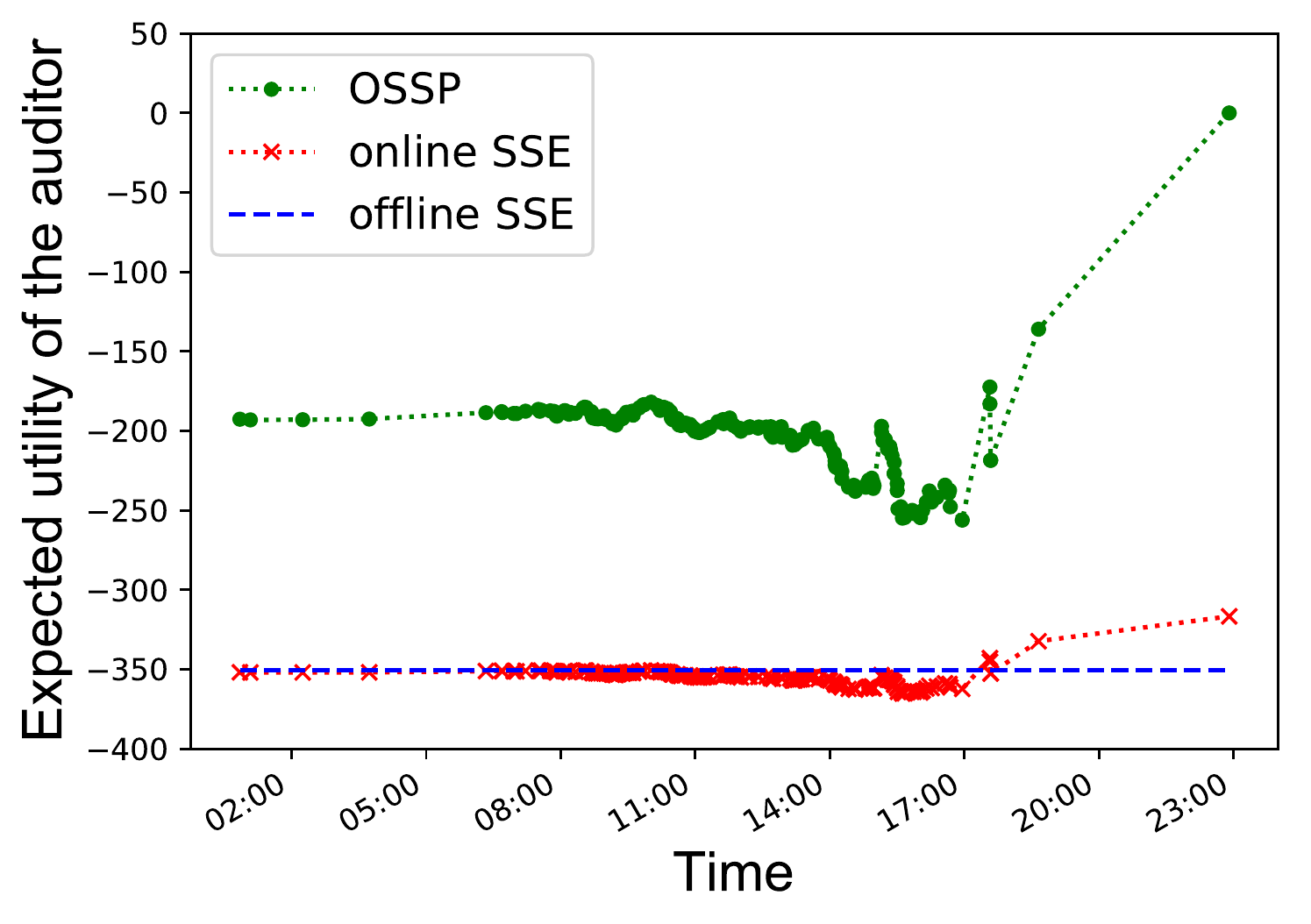}}%
\subfigure[ Day 12]{%
\label{fig_mul_12}%
\includegraphics[width=4.0cm]{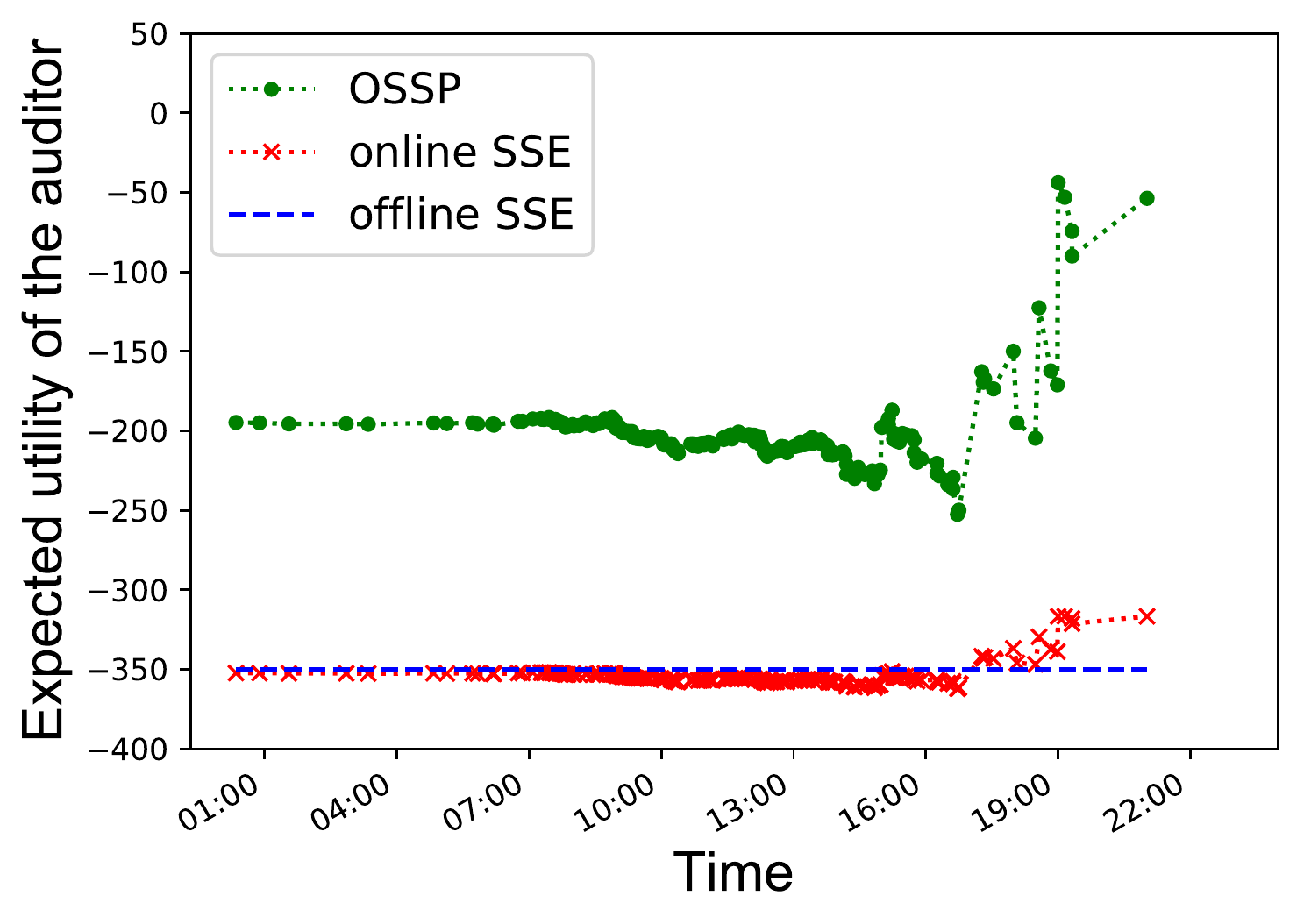}}%

\subfigure[ Day 13]{%
\label{fig_mul_13}%
\includegraphics[width=4.0cm]{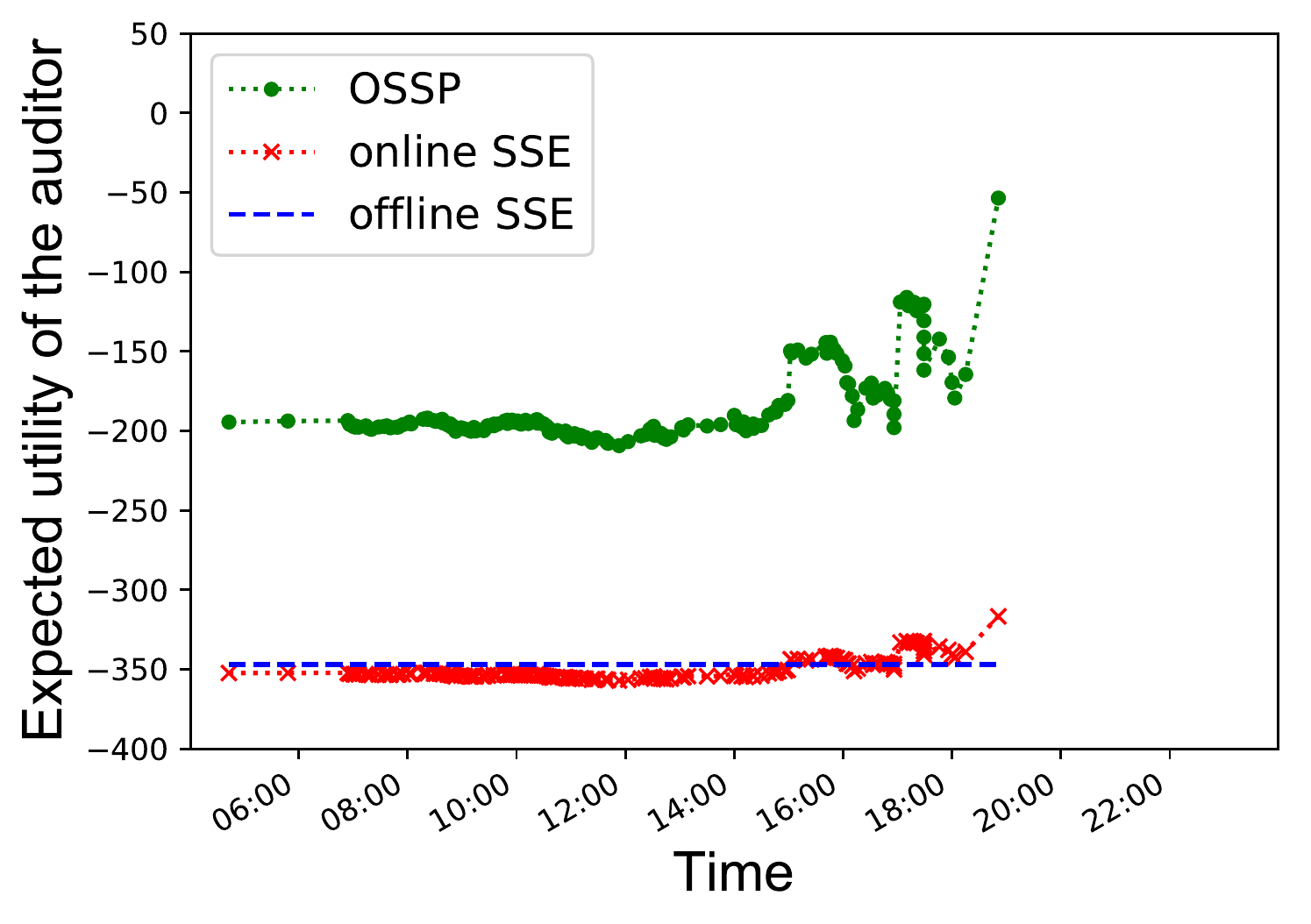}}%
\subfigure[ Day 14]{%
\label{fig_mul_14}%
\includegraphics[width=4.0cm]{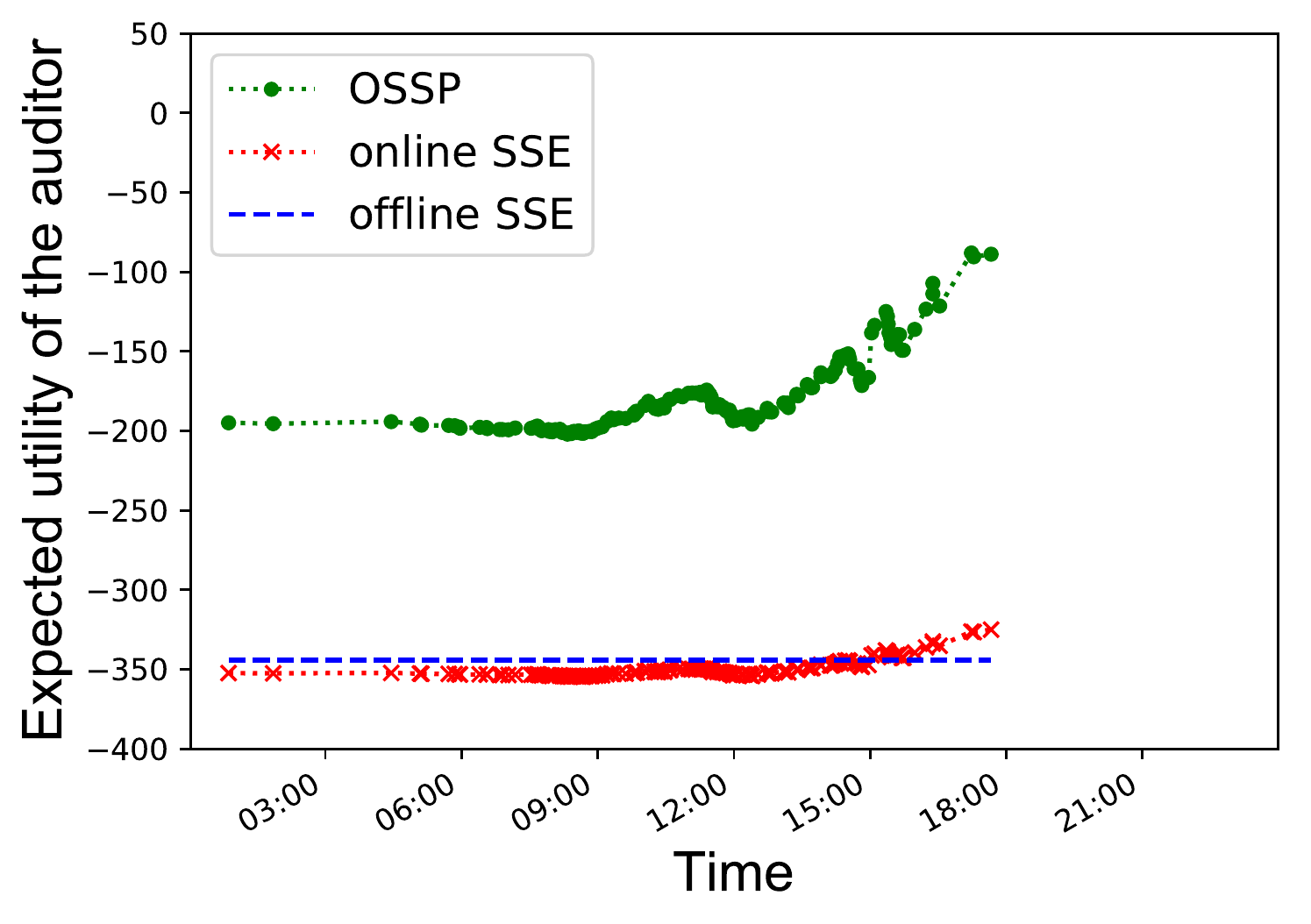}}%
\subfigure[ Day 15]{%
\label{fig_mul_15}%
\includegraphics[width=4.0cm]{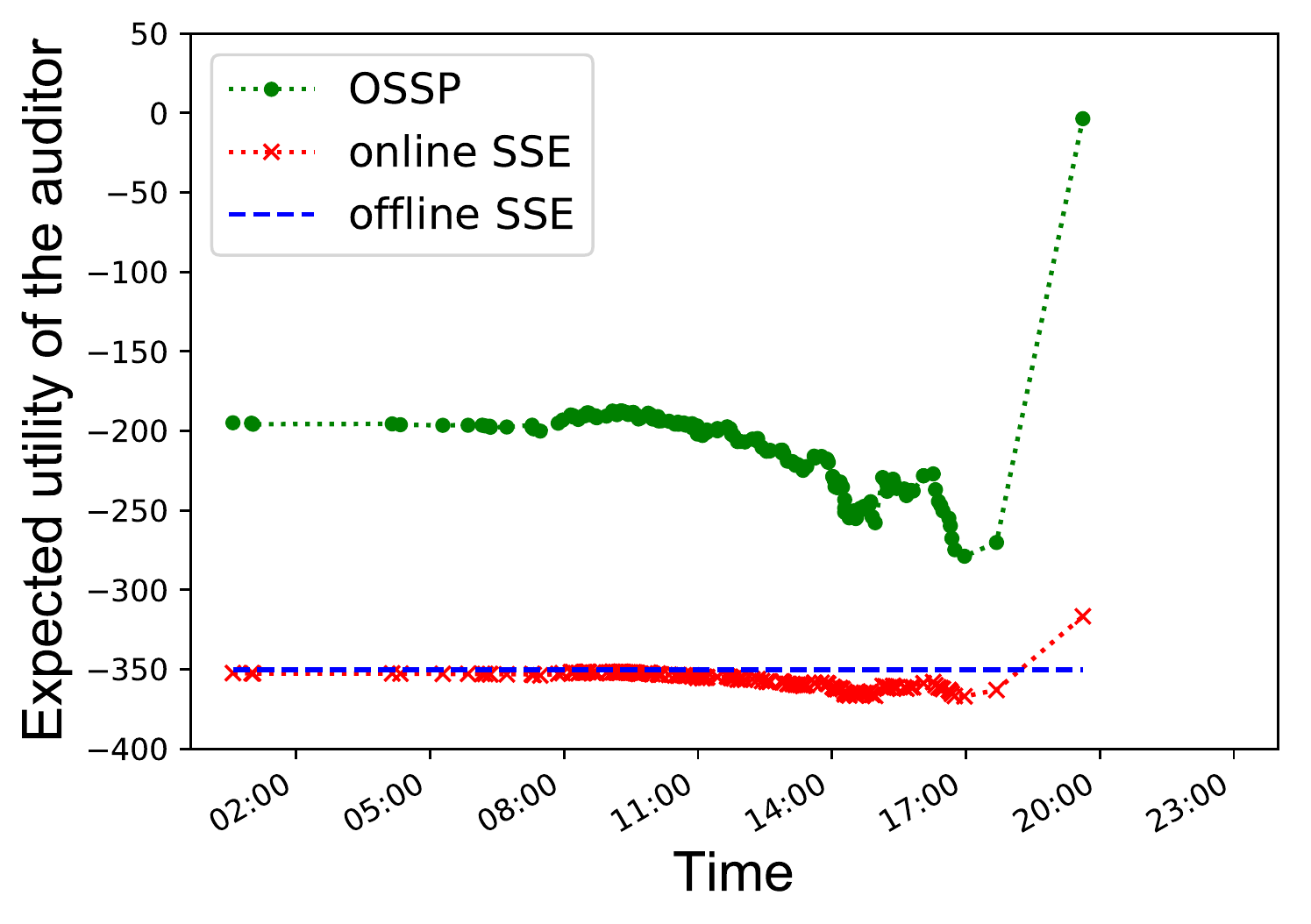}}%

\caption{The auditor's expected utility in the OSSP and alternative equilibria for the 7 alert types with a total budget of $B = 50$. We applied $\alpha = 1\%$ and $C_t=-1$ for the OSSP.}
\vspace{-4mm}
\end{figure*}

The audit cycle is defined as one day from \emph{0:00:00} to \emph{23:59:59}. From the dataset, we construct $15$ groups, each of which contains the alert logs of $41$ continuous normal working days as the historical data (for estimating the distributions of future alerts in all types), and the alert logs of the $1$ subsequent day as the day for testing purpose. We set up a real time environment for evaluating the performance in terms of the auditor's expected utility. We set the audit cost per alert to $V^t = 1, \forall t \in \{1,...,|T|\}$. From the alert logs of three months, we obtain the frequency at which users quit when they receive the warning messages in our dataset. According to this observation, in our experiments we set the probability of quitting as $P^{t} = 0.186$ in the SAG model for all types.

We compare the real time auditor's expected utility for each triggered alert between the OSSP (the optimal objective value of LP \eqref{basic_game}) and both the \emph{offline} and \emph{online SSE} (the optimal objective value of LP \eqref{SSG_SSE}). The offline SSE corresponds to the traditional auditing method, which determines the auditing strategy for each alert at the end of the auditing cycle. By contrast, the online SSG determines the auditing strategy for each alert in real time, which is equivalent to an SAG without signaling. 

One significant challenge in comparing the OSSP with the online SSE is that the real time budget consumption in the SAG is determined by the sampling result of warning/no warning and, thus is not deterministic. This leads to a situation where, for the time series of alerts in each audit cycle, if there is no intervention, then the online SSG and the SAG will move independently with respect to the game status. As such, their performance cannot be directly compared.
To set up a well-controlled environment for comparison, for each incoming alert we focus on the online SSG with its game status be the same as the current SAG instance. Recall that for the SAG, the auditor needs to reserve a portion of the total auditing budget for inspecting the repeated data requests at the end of each audit cycle. Due to the fact that it is unnecessary for the online SSG, we set the available budget at each incoming alert in the online SSG  to be equal to  the sum of 1) the available budget of the SAG instance at the current time point, and 2) the reserved budget of the SAG for the additional inspection of the repeated data requests. By doing so, it makes our comparison reasonable and fair.

To investigate the robustness of the results, we evaluate the performance by varying three factors. First, we vary the loss value for the auditor with respect to each quit of a normal user when receiving a warning message. We set $C_{t} = \{-1, -5, -10\}$.\footnote{To the best of our knowledge, there is no perfect measure for this loss in the EMR application domain.}
Second, to deter the attacker who  quits until they receive no warning in the safe period for an SAG (where $p_0^{t*}=0$ as shown in Theorem \ref{prob_0}), we assess a series of constant budgets, which we set to $\alpha = \{1\%, 5\%\}$ of the total available budget $B$. We do not consider this situation in the baseline strategies because such loss does not apply. Third, we vary the total auditing budget. Specifically, we consider $B = \{30, 50, 70\}$. By setting $B = 50$, the available budgets for the SAG at the very beginning time point of an audit cycle are $49.5$ for $\alpha = 1\%$ and $47.5$ for $\alpha = 5\%$, respectively.

\subsection{Results} 
We considered all $7$ alert types described in Table \ref{type}.  Due to space limitations, we only show the sequential results of 15 sequential testing days along the timeline in Figures \ref{fig_mul_1}-\ref{fig_mul_15} by applying $B = 50, C_{t} = -1$ for all types and $\alpha = 1\%$.    
 
 \begin{table*}[ht]
\center
\fontsize{7.5}{9}\selectfont
\caption{The advantages of OSSP over online SSE in terms of the mean (and the standard deviation) of the differences in the auditor's expected utility (15 testing days). }
\label{diff}
\begin{tabular}{@{}ccccccc@{}}
\toprule
\midrule
 \multirow{2}{*}{$B$}    & \multicolumn{2}{c}{$C_t=-1$} & \multicolumn{2}{c}{$C_t=-5$} & \multicolumn{2}{c}{$C_t=-10$} \\ \cmidrule(l){2-3} \cmidrule(l){4-5} \cmidrule(l){6-7}
                      & $\alpha =1\%$  & $\alpha =5\%$  & $\alpha =1\%$  & $\alpha =5\%$ &  $\alpha =1\%$  & $\alpha =5\%$ \\
\midrule
  $30$ &     $60.87 \pm 28.31~~15.99\%$    &   $47.01 \pm 32.17~~12.45\%$      &    $40.43 \pm 23.95~~10.59\%$    &   $29.89 \pm 28.77~~7.92\%$       &   $26.91 \pm 25.77~~7.06\%$     &    $10.94 \pm 24.93~~2.90\%$          \\[0.2em]
  $50$ &   $165.83 \pm 24.49~~47.26\%$      &    $147.51 \pm 27.74~~42.65\%$        &    $143.19 \pm 33.98~~40.87\%$    &   $117.52 \pm 34.56~~34.20\%$      &  $127.31 \pm 37.55~~36.23\%$      &    $106.21 \pm 38.85~~31.21\%$          \\[0.2em]
  $70$ &   $252.57 \pm 20.44~~77.31\%$      &   $235.14 \pm 23.57~~72.87\%$        &    $227.59 \pm 33.10~~69.31\%$    &  $204.33 \pm 36.77~~63.63\%$         &   $225.35 \pm 37.58~~68.73\%$     &    $198.69 \pm 40.93~~61.89\%$        \\
\bottomrule     
\end{tabular}
\vspace{-2mm}
\end{table*}

It is noteworthy that the type for each alert may not be aligned with the optimal attacking type in the OSSP strategy. Thus, to compare the approaches, we only apply the SAG on alerts whose type is equal to the best attacking type in the OSSP. For alerts whose types differ, we simply apply the online SSE strategy and use its optimal coverage probability to update the real time available budget. When applying SAGs, we first optimize the signaling scheme, then randomly sample whether to send a warning according to $\mathbf{P}(\xi^{\tau}_1)$. Next, we update (in real time) the available budget based on the signal.

Figures \ref{fig_mul_1}-\ref{fig_mul_15} illustrate the real time expected utility of the auditor. It can be seen that the majority of alerts were triggered between $8:00 AM$ and $5:00 PM$, which generally corresponds to the normal working hours of VUMC.  After this period, the rate of alerts slows down considerably. Note that  the trend for offline SSE is flat because, in this method, the auditor's expected utility is the same for each alert regardless of when it is triggered. 

There are several notable findings and implications.
First, in terms of the expected utility of the auditor, OSSP significantly outperforms the offline SSE and the online SSE. 
This suggests that the SAG increases auditing effectiveness. We believe that this advantage is due to the optimized signaling mechanism, which ensures the loss of the auditor is zero when sending warning messages.
Second, at the end of each testing day, the auditor's expected utility for each approach does not drop below the online SSE. We believe that this is an artifact of the knowledge rollback, which slows down the budget consumption in this period. In particular, at the end of multiple testing days, such as illustrated in Figures \ref{fig_mul_1}, \ref{fig_mul_6}, \ref{fig_mul_7}, \ref{fig_mul_8}, \ref{fig_mul_10}, \ref{fig_mul_11} and \ref{fig_mul_15}, the expected auditor loss approaches $0$. Third, the sequences of online SSE are close to the corresponding offline SSE sequences. This indicates that the auditing procedure does not benefit from determining only the coverage probability for each of the alert types in real time. In other words, the signaling mechanism in the SAG can assist the auditing tasks in various environments. Moreover, the advantage of OSSP over online SSE grows with the overall budget. 

We expanded the investigation to consider various conditions of the auditing tasks. We computed the mean (and standard deviation\footnote{Note that the distributions are not necessarily Gaussian. The standard deviations are largely dominated by the ending periods of testing days, where the expected utility of the auditor in the OSSP is usually close to $0$.}) differences between the OSSP and the corresponding online SSE for each triggered alert across $15$ testing days by varying the total auditing budget, the loss of the auditor on each quit of normal users, and the percentage of the budget for inspecting anomalous repeated requests. 
The results are shown in Table \ref{diff}, where we also indicate the percentage of the averaged improvement in each setting. Here, this value is defined as the absolute improvement on the expected utility of the auditor divided by the optimal auditor's expected utility in the online SSE. From the results, we have the following significant observations. First, it is notable that OSSP consistently outperforms the online SSE with respect to the auditor's expected utility in a variety of auditing settings. For example, in the setting that $C_t = -1$ for all $t$ and $\alpha=1\%$, as $B$ grows from $30$ to $70$, the auditor's expected utility improvement grows from $16\%$ to $77\%$. This is a trend that holds true for other settings as well. Second, by fixing $B$ and $C_t$ for all $t$, the auditor's expected utility decreases when we reserve more budget to investigate the repeated requests by single user. Yet, this is not unexpected because this approach reduces the amount of consumable auditing resources. Third, by increasing the cost of deterring a single normal data request, we also weaken the advantages of OSSP over the online SSE (when $B$ and $\alpha$ are held constant).

\begin{table}[ht]
\center
\fontsize{7.5}{9}\selectfont
\caption{The mean and standard deviation of auditor's expected utility at OSSP  as a function of $P^t$ (15 testing days). }
\label{P_scare}
\begin{tabular}{@{}cccc@{}}
\toprule
\midrule
 \multirow{2}{*}{$C_t$}    & \multicolumn{3}{c}{$P^t$ for all $t$} \\ \cmidrule(l){2-4} 
                      & $\times 1.0$  & $\times0.5$  & $\times0.1$ \\
\midrule
 $-1$  & $-185.54\pm29.85$ & $-179.92\pm32.71$ & $-175.47\pm32.97$  \\[0.2em]
  $-5$ & $-208.60\pm41.91$ & $-201.34\pm33.92$ & $-180.85\pm32.75$  \\[0.2em]
\bottomrule     
\end{tabular}
\vspace{-2mm}
\end{table}

Next, we considered how the probability of being scared away for normal users (i.e., $P^t$) influences the auditor's expected utility. Recall that, in the experiments reported on so far, we adopted $P^t = 0.186$, an estimate based on an environment that relied upon an unoptimized signaling procedure. However, this value can change in practice for several reasons. First, an optimized signaling scheme will likely influence users' access patterns, such as the frequency of triggering alerts, as well as how users respond to a signaling mechanism. Second, the probability $P^t$ can decrease, if an organization effectively performs policy training with its employees, such that normal users may be less likely to be scared away if they receive a warning message when requesting access to a patient's record. Table \ref{P_scare} shows the expected utility of the auditor at OSSP by varying the input of $P^t$ in the setting of $B=50$. We apply three values of $P^t$ by reducing the original value to its $100\%, 50\%$ and $10\%$. It can be seen that the auditor's expected utility under OSSP improves as $P^t$ reduces. When holding $C_t$ constant, a $t$-test reveals that each pair of performances is statistically significantly different with $p < 10^{-6}$. This indicates that reducing the frequency of quitting for normal users directly reduces the usability costs and, thus, improves the auditing efficiency.

In addition, we tested the average running time for optimizing the SAG on a single alert across all the testing days. Using a laptop running Mac OS, an Intel i7 @ 3.1GHz, and 16GB of memory, we observed that the SAG could be solved in $0.06$ seconds on average. As a consequence, it is unlikely that system users would unlikely perceive the extra processing time associated with optimizing the SAG in practice.


%% file: II_literature.tex
\section{Related Work}
There have been a number of investigations into effective alert management strategies and efficient auditing mechanisms for database systems. In this section, we review the game-theoretic developments that are related to our investigation.

Blocki \emph{et al.} first modeled the audit problem between an auditor and an auditee as a classic security game. In this setting, players act strategically and the goal is to learn an optimal resource allocation strategy that optimizes the expected payoff of the auditor \cite{blocki2013audit}. To simulate the real audit environment, Blocki \emph{et al.} generalized the framework by accounting for the situations with multiple defender resources \cite{blocki2015audit}. However, their methods treat alerts as a set of existing targets that could be attacked, a modeling decision that cannot be readily generalized into the database audit setting. To solve this challenge, Yan \emph{et al.} introduced a game theoretic audit approach to 1) prioritize the order in which types of alerts are investigated and 2) provide an upper bound on how much resource to allocate for auditing each type \cite{yan2018get,yan2019database}. Schlenker \emph{et al.} introduced an approach dealing with how   to assign alerts to security analysts was proposed, where each analyst has different areas of expertise . However, all of these investigations adopted a classic security game framework, which, as our experiments show, hinder the  efficacy of the system.

It has been shown that the integration of a  signaling mechanism into adversarial settings can improve protection. In particular, Xu \emph{et al.} proposed a two-stage security game model to protect targets with a better performance. In the first stage, the defender allocates inspection resources and the attacker selects a target. In the second stage, the defender reveals information, potentially deterring the attacker's attack plan of attack \cite{xu2015exploring}. The advantages of signaling were subsequently extended to Bayesian Stackelberg games, where players have payoff-relevant private information \cite{dughmi2016algorithmic}. It has been shown that signaling also boosts defensive performance in security games, specifically for the task of assigning randomized human patrollers and sensors to protect important targets \cite{xu2018strategic}. However, these investigations aimed to protect existing physical targets as well. The methodology does not easily fit into the auditing environment, where the timing of budget assignment and signaling are reversed.

%% file: VII_discussion.tex
\section{Discussion}

In this paper, we integrated signaling into auditing frameworks. We strategically warn the attacker in real time and then realize the audit strategy at the end of the audit cycle with an offline mode. In particular, we formalized the usability cost in our approach to model the real-world audit scenario. We further illustrated that such a defensive strategy improves the performance of defenders over existing game theoretic alternatives using real EMR auditing data. Our framework is generalizable to more powerful attackers because as long as the adversarial behavior can be represented by pattern(s), it will fit into our model. As such, our audit model is applicable to any capability of the attacker.

There are several limitations we wish to highlight as opportunities for future investigations. First, since this was a pilot study  for SAGs, we assumed there is only one attacker and that they had a fixed payoff structure in each audit cycle. However, in practice, there may exist many types of attackers. As a next step, we believe that the SAG can be extended for a Bayesian setting where the payoff structure  of the attacker varies according to types. Second, the single attacker assumption is due to the case that the number of malicious employees within an HCO is expected to be very small, such that it is rare that two attackers will be realized in the same audit cycle. Still, this is not unheard of and, thus, it is necessary to address the situation of multiple attackers. Third, in this paper, we assumed that the attacker is perfectly rational. This is a strong assumption and may lead to an unexpected loss in practice. Thus, a more robust version of the SAG will be needed for wide deployment. Fourth, the scalability of solving the SAG, with respect to the number of alert types, needs more investigation in future.

\section{Conclusion}
Alert-based auditing is often deployed in database systems to address a variety of attacks to the data resources being stored and processed. However, the  volume of alerts is often beyond the capability of administrators, thus limits the effectiveness of auditing. Our research illustrates that strategically incorporating signaling mechanisms into the data request workflow can significantly improve the auditing work. We investigated the features, as well as, the value of a game theoretic Signaling Audit Game, along with an Online Stackelberg Signaling Policy to solve the game. While we demonstrated the feasibility of this approach with the audit logs of an electronic medical record system at a large academic medical center, the approach is sufficiently generalized to support auditing in a wide range of environments.  Though our investigation illustrates the merits of this approach, there are certain limitations  that provide opportunities for extension and hardening of the framework for real world deployment.